\documentclass[prodmode,acmtalg]{acmsmall}

\usepackage{todonotes}
\usepackage{cleveref}
\usepackage[draft]{fixme}

\fxsetup{mode=multiuser,theme=color, layout=margin}
\FXRegisterAuthor{gn}{agn}{\color{red} \bf Gonzalo}
\FXRegisterAuthor{np}{anp}{\color{green} \bf  Nicola}
\FXRegisterAuthor{tk}{atk}{\color{blue} \bf Tomasz}

\newcommand{\dd}{\mathinner{.\,.}}
\newcommand{\sub}{\subseteq}
\newcommand{\sm}{\setminus}
\newcommand{\Parsing}{\mathcal{P}}

\usepackage{graphicx}
\usepackage{color}
\usepackage{amsmath}
\usepackage{amssymb}
\usepackage[noend,ruled,linesnumbered]{algorithm2e}

\crefname{algocf}{Algorithm}{Algorithms}
\crefname{equation}{Eq.}{Eq.}

\newtheorem{observation}[theorem]{Observation}
\newtheorem{claim}[theorem]{Claim}

\begin{document}

\markboth{M. B. Ettienne et al.}{Optimal-Time Dictionary-Compressed Indexes}

\title{Optimal-Time Dictionary-Compressed Indexes}

\author{%
Anders Roy Christiansen
\affil{The Technical University of Denmark, Denmark}
Mikko Berggren Ettienne\affil{The Technical University of Denmark, Denmark}
Tomasz Kociumaka\affil{Bar-Ilan University, Israel, and University of Warsaw,
Poland}
Gonzalo Navarro\affil{CeBiB and University of Chile, Chile}
Nicola Prezza\affil{University of Pisa, Italy}
}

\begin{abstract}
We describe the first self-indexes able to count and locate pattern
occurrences in optimal time within a space bounded by the size of the most
popular dictionary compressors. To achieve this result we combine several
recent findings, including \emph{string attractors} --- new combinatorial
objects encompassing most known compressibility measures for highly repetitive
texts ---, and grammars based on \emph{locally-consistent parsing}. 
	
More in detail, let $\gamma$ be the size of the smallest attractor for a text $T$ of length $n$. The measure $\gamma$ is an (asymptotic) lower bound to the size of dictionary compressors based on Lempel--Ziv, context-free grammars, and many others. The smallest known text representations in terms of attractors use space $O(\gamma\log(n/\gamma))$, and
our lightest indexes work within the same asymptotic space. Let $\epsilon>0$ be a suitably small constant fixed at construction time, $m$ be the pattern length, and $occ$ be the number of its text occurrences. Our index  counts pattern occurrences in $O(m+\log^{2+\epsilon}n)$ time, and locates them in $O(m+(occ+1)\log^\epsilon n)$ time. 
These times already outperform those of most dictionary-compressed indexes, while obtaining the least asymptotic space for any index searching within $O((m+occ)\,\textrm{polylog}\,n)$ time.
Further, by increasing the space to $O(\gamma\log(n/\gamma)\log^\epsilon n)$, we reduce the
locating time to the optimal $O(m+occ)$, and within
$O(\gamma\log(n/\gamma)\log n)$
space we can also count in optimal $O(m)$ time. No dictionary-compressed index had obtained this time before. All our indexes can be
constructed in $O(n)$ space and $O(n\log n)$ expected time.

As a byproduct of independent interest, we show how to build, in $O(n)$ expected time and without knowing the size $\gamma$ of the smallest attractor (which is NP-hard to find), a run-length context-free grammar of size $O(\gamma\log(n/\gamma))$ generating (only) $T$.
As a result, our indexes can be built without knowing $\gamma$.
\end{abstract}

\begin{bottomstuff}
Kociumaka supported by ISF grants no. 824/17 and 1278/16 and by an ERC grant
MPM under the EU's Horizon 2020 Research and Innovation Programme (grant no.
683064).
Navarro supported by Fondecyt grant 1-170048, Chile; Basal Funds FB0001,
Conicyt, Chile.
Prezza supported by the project MIUR-SIR CMACBioSeq (``Combinatorial methods
for analysis and compression of biological sequences'') grant no.~RBSI146R5L.

A preliminary  version of this article appeared in {\em Proc. LATIN'18} \cite{CE18}.

Author's addresses:
Anders Roy Christiansen, The Technical University of Denmark, Denmark,
{\tt aroy@dtu.dk}. 
Mikko Berggren Ettienne,  The Technical University of Denmark, Denmark,
{\tt miet@dtu.dk}. 
Tomasz Kociumaka, Department of Computer Science, Bar-Ilan University, Ramat
Gan, Israel \and Institute of Informatics, University of Warsaw,
Poland, {\tt kociumaka@mimuw.edu.pl}. 
Gonzalo Navarro, CeBiB -- Center for Biotechnology and Bioengineering, Chile
\and Department of Computer Science, University of Chile, Chile, {\tt
gnavarro@dcc.uchile.cl}. 
Nicola Prezza, Department of Computer Science, University of Pisa, Italy,
{\tt nicola.prezza@di.unipi.it}. 
\end{bottomstuff}



\maketitle


\section{Introduction}

The need to search for patterns in large string collections lies at the heart of many text retrieval, analysis, and mining tasks, and techniques to support it efficiently have been studied for decades: the suffix tree, which is the landmark solution, is over 40 years old \cite{DBLP:conf/focs/Weiner73,DBLP:journals/jacm/McCreight76}. The recent explosion of data in digital form led the research since 2000 towards {\em compressed self-indexes}, which support text access and searches within compressed space \cite{DBLP:journals/csur/NavarroM07}. This research, though very successful, is falling short to cope to a new wave of data that is flooding our storage and processing capacity with volumes of higher orders of magnitude that outpace Moore's Law \cite{Plos15}. Interestingly enough, this massive increase in data size is often not accompanied with a proportional increase in the amount of information that data carries: much of the fastest-growing data is {\em highly repetitive}, for example thousands of genomes of the same species, versioned document and software repositories, periodic sky surveys, and so on. Dictionary compression of those datasets typically reduces their size by two orders of magnitude \cite{GNP18}. Unfortunately, previous self-indexes build on statistical compression, which is unable to capture repetitiveness \cite{KN13}; therefore, a new generation of compressed self-indexes based on dictionary compression is emerging.

 Examples of successful compressors from this family include (but are not limited to) the Lempel--Ziv factorization~\cite{LZ76}, of size $z$; context-free grammars~\cite{KY00} and run-length context-free grammars~\cite{DBLP:conf/mfcs/NishimotoIIBT16}, of size $g$; bidirectional macro schemes
\cite{SS82}, of size $b$; and collage systems \cite{KidaMSTSA03}, of size $c$. Other compressors that are not dictionary-based but also perform well on repetitive text collections are the run-length Burrows--Wheeler 
transform~\cite{BWT}, of size $\rho$, and the CDAWG~\cite{blumer1987complete}, of size $e$.
A number of compressed self-indexes have been built on top of those
compressors; \citeN{GNP18} give a thorough review. 

Recently, \citeN{KP18} showed that all the above-mentioned
repetitiveness measures (i.e., $z$, $g$, $b$, $c$, $\rho$, $e$) are never asymptotically smaller than the size $\gamma$ of a new
combinatorial object called \emph{string attractor}. This and subsequent
works~\cite{KP18,NP18,prezza2019optimal} showed that efficient access and searches can be
supported within $O(\gamma\log(n/\gamma))$ space. By the nature of this new
repetitiveness measure, such data structures are universal, in the sense that
they can be used on top of a wide set of dictionary-compressed 
representations of $T$.

\paragraph{Our results}
In this article we obtain the best results on attractor-based indexes, including
the first optimal-time search complexities within space bounded in terms of $\gamma$, $z$, $g$, $b$,
or $c$. We combine and improve upon three recent results:

\begin{enumerate}
\item \citeN[Thm.~2]{NP18} presented the first index that builds on
an attractor of size $\gamma$ of a text $T[1..n]$. It uses 
$O(\gamma\log(n/\gamma))$ space and finds the $occ$ occurrences of a pattern 
$P[1..m]$ in time $O(m\log n + occ (\log\log(n/\gamma)+\log^\epsilon \gamma))$ 
for any constant $\epsilon>0$.
\item \citeN[Thm.~2(3)]{CE18} presented an index that builds on
the Lempel--Ziv parse of $T$, of $z \ge \gamma$ phrases, which uses 
$O(z\log(n/z))$ space and searches
in time\footnote{This is the conference version of the present article, where we mistakenly claim a slightly better time of $O(m + \log^\epsilon z + occ(\log\log n + \log^\epsilon z))$. The error can be traced back to the wrong claim that our two-sided range structure, built on $O(z\log(n/z))$ points, answers queries in $O(\log^\epsilon z)$ time (the correct time is, instead, $O(\log^\epsilon (z\log(n/z)))$). The second occurrence of $\log^\epsilon z$, however, is correct, because the missing term is absorbed by $O(\log\log n)$.} $O(m + \log^\epsilon(z\log(n/z)) + occ(\log\log n + \log^\epsilon z))$.
\item \citeN[Thm.~5]{Nav18} presented the first index that builds on the 
Lempel--Ziv parse of $T$ and counts the number of occurrences of $P$ in $T$
(i.e., computes $occ$) in time $O(m\log n + m\log^{2+\epsilon} z)$, using $O(z\log(n/z))$ space.
\end{enumerate}

Our contributions are as follows:

\begin{enumerate}
\item
We obtain, in space $O(\gamma\log(n/\gamma))$, an index that lists all the
occurrences of $P$ in $T$ in time $O(m + \log^\epsilon \gamma + 
occ \log^\epsilon (\gamma\log(n/\gamma)))$,
thereby obtaining the best space and improving the time from previous works \cite{CE18,NP18}. 
\item
We obtain, in space $O(\gamma\log(n/\gamma))$, an index that counts the
occurrences of $P$ in $T$ in time $O(m+\log^{2+\epsilon} 
(\gamma\log(n/\gamma)))$, which outperforms the previous result \cite{Nav18} 
both in time and space. 
\item
Using more space, $O(\gamma\log(n/\gamma)\log^\epsilon n)$, we list the
occurrences in optimal $O(m+occ)$ time, and within space
$O(\gamma\log(n/\gamma)\log n)$, we count them in optimal  
$O(m)$ time. 
\end{enumerate}

We can build all our structures in $O(n\log n)$ expected time and $O(n)$ working
space, without the need to know the size $\gamma$ of the smallest attractor.

Our first contribution uses the minimum known asymptotic space,
$O(\gamma\log(n/\gamma))$, for any dictionary-compressed index searching in
time $O((m+occ)\,\mathrm{polylog}\,n)$ \cite{GNP18}. Only recently
\cite{NP18}, it has been shown that it is possible to search within this
space. Indeed, our new index outperforms most dictionary-compressed indexes,
with a few notable exceptions like \citeN{gagie2014lz77}, who use
$O(z\log(n/z))$ space and $O(m\log m + occ \log\log n)$ search time (but,
unlike us, assume a constant alphabet), and \citeN{PhBiCPM17}, who use  $O(z\log(n/z)\log\log z)$ space and $O(m+occ \log\log n)$ search time without making any assumption on the alphabet size.
Our second contribution lies on a less explored area, since the first index able to count efficiently within dictionary-bounded space is very recent \cite{Nav18}.

Our third contribution yields the first indexes with space bounded in terms of $\gamma$, $z$, $g$, $b$, or
$c$, multiplied by any $O(\textrm{polylog}\,n)$, that searches in optimal time. Such optimal times have been obtained, instead, by using 
$O(\rho\log(n/\rho))$ space \cite{GNP18}, or using $O(e)$ space \cite{BC17}. Measures $\rho$ and $e$, however, are not related to dictionary compression and, more importantly, have no known useful upper bounds in terms of $\gamma$. Further, experiments \cite{BCGPR15,GNP18} show that they are usually considerably larger than $z$ on repetitive texts.

As a byproduct of independent interest, we show how to build a run-length context-free grammar (RLCFG) of size 
$O(\gamma\log(n/\gamma))$ generating (only) $T$, where $\gamma$ is the size of 
the smallest attractor, in $O(n)$ expected time {\em and without the need to 
know the attractor}. We use this result to show that our indexes do not need to know an attractor, nor its minimum possible size $\gamma$ (which is NP-hard to obtain \cite{KP18}) in order to achieve their attractor-bounded results. This makes our results much more practical.
Another byproduct is the generalization of our results to arbitrary CFGs and, especially, RLCFGs, yielding slower times in $O(g)$ space, which can potentially be $o(\gamma\log(n/\gamma))$.

\paragraph{Techniques}
A key component of our result is the fact that one can build a locally-consistent and 
locally-balanced grammar generating (only) $T$ such that only a few splits of
a pattern $P$ must be considered in order to capture all of its ``primary''
occurrences \cite{KU96}. Previous parsings had obtained $O(\log m \log^* n)$
\cite{NIIBT15} and $O(\log n)$ \cite{DBLP:conf/soda/GawrychowskiKKL18} splits,
but now we build on a parsing by \citeN{DBLP:journals/algorithmica/MehlhornSU97} to obtain $O(\log m)$ splits with a grammar of size $O(\gamma\log(n/\gamma))$.

Our first step is to define a variant of Mehlhorn et al.'s randomized parsing and prove, in Section~\ref{sec:lcp}, that it enjoys several locality properties we require later for indexing. In Section~\ref{sec:locality}, we use the parsing to build a RLCFG with the local balancing and local consistency properties we need. We then show, in Section~\ref{sec:attractor}, that the size of this grammar is bounded by $O(\gamma\log(n/\gamma))$, by proving that new nonterminals appear only around attractor positions. 

In that section, we also show that the grammar can be built without knowing the minimum size $\gamma$ of an attractor of $T$. This is important because, unlike $z$, which can be computed in $O(n)$ time, finding $\gamma$ is NP-hard \cite{KP18}. For this sake we define a new measure of compressibility, $\delta \le \gamma$, which can be computed in $O(n)$ time and can be used to bound the size of the grammar. 

Section~\ref{sec:index} describes our index.
We show how to parse the pattern in linear time using the same text grammar, and how to do efficient substring extraction and Karp--Rabin fingerprinting from a 
RLCFG. Importantly, we prove that only $O(\log m)$ split points are necessary in our grammar. All these elements are needed to obtain time linear in $m$. We also build on existing techniques \cite{CNspire12} to obtain time linear in $occ$ for the ``secondary'' occurrences; the primary ones are found in a two-dimensional data structure and require more time. Finally, by using a larger two-dimensional structure and introducing new techniques to handle short patterns, we raise the space to $O(\gamma\log(n/\gamma)\log^\epsilon n)$ but obtain the first dictionary-compressed index using optimal $O(m+occ)$ time.

In Section~\ref{sec:counting} we 
use the fact that only $O(\log m)$ splits must be considered to reduce the counting time of 
\citeN{Nav18}, while making its space attractor-bounded as well. This requires handling the run-length rules of RLCFGs, which turns out to require new ideas exploiting string periodicities. Further, by handling short patterns separately and raising the space to $O(\gamma\log(n/\gamma)\log n)$, we obtain the first dictionary-compressed index that counts in optimal time, $O(m)$.
 
 Along the article we obtain various results on accessing and indexing specific RLCFGs. We generalize them to arbitrary CFGs and RLCFGs in Appendix~\ref{sec:rlcfgs}.
 
\medskip
An earlier version of this article appeared in {\em Proc. LATIN'18} \cite{CE18}. This article is an  exhaustive rewrite where we significantly extend and improve upon the conference results. We use a slightly different grammar, which requires re-proving all the results, in particular correcting and completing many of the proofs in the conference paper. We have also reduced the space by building on attractors instead of Lempel--Ziv parsing, used better techniques to report secondary occurrences and handle short patterns, and ultimately obtained optimal locating time.
All the results on counting are also new.


\section{Basic Concepts} \label{sec:basics}

\paragraph*{Strings and texts}
A {\em string} is a sequence $S[1\dd \ell] = S[1] S[2] \cdots S[\ell]$ of {\em
symbols}. The symbols belong to an {\em alphabet} $\Sigma$, which is a finite 
subset of the integers. The {\em length} of $S$ is written as
$|S|=\ell$. 

A string $Q$ is a \emph{substring} of $S$ if $Q$ is empty or $Q=S[i] \cdots S[j]$ 
for some indices $1\le i \le j \le \ell$.
The \emph{occurrence} of $Q$ at position $i$ of $S$ is a \emph{fragment} of $S$
denoted $S[i\dd j]$. We then also say that $S[i\dd j]$ \emph{matches} $Q$.
We assume implicit casting of fragments to the underlying substrings 
so that $S[i\dd j]$ may also denote $S[i]\cdots S[j]$ in contexts requiring strings rather than fragments.

A  {\em suffix} of $S$ is a fragment of the form $S[i\dd \ell]$, and a {\em prefix}
is a fragment of the form $S[1\dd i]$. 
The juxtaposition of strings and/or 
symbols represents their concatenation, and the exponentiation denotes the
iterated concatenation. The {\em reverse} of $S[1\dd \ell]$ is $S^{rev} = S[\ell] S[\ell-1]
\cdots S[1]$.

 We will index a string $T[1\dd n]$, called the {\em text}. We assume our text
 to be flanked by special symbols $T[1]=\#$ and $T[n]=\$$ that belong to 
 $\Sigma$ but occur nowhere else in $T$. This, of course, does not change any 
 of our asymptotic results, but it simplifies matters. 

\paragraph*{Karp--Rabin signatures}
{\em Karp--Rabin fingerprinting} \cite{DBLP:journals/ibmrd/KarpR87} assigns to every string $S[1\dd \ell]$ a
signature $\kappa(S) = (\sum_{i=1}^\ell S[i]\cdot c^{i-1}) \bmod \mu$ for
a suitable integers $c$ and a prime number $\mu$. It is possible to build a signature
formed by a pair of functions $\langle \kappa_1,\kappa_2 \rangle$ 
guaranteeing no collisions between substrings of $T[1\dd n]$, in 
$O(n\log n)$ expected time \cite{DBLP:journals/jda/BilleGSV14}.

\paragraph*{With high probability}
The term \emph{with high probability} (\emph{w.h.p.}) means with probability 
at least $1-n^{-c}$ for an arbitrary constant parameter $c$,
where $n$ is the input size (in our case, the length of the text).

\paragraph*{Model of computation}
We use the RAM model with word size $w = \Omega(\log n)$, allowing classic arithmetic
and bit operations on words in constant time.
Our logarithms are to the base $2$ by default.


\section{Locally-Consistent Parsing} \label{sec:lcp}

A string $S[1\dd n]$ can be parsed in a \emph{locally 
consistent} way, meaning that equal substrings are largely parsed in the same 
form.  We use a variant of the parsing of \citeN{DBLP:journals/algorithmica/MehlhornSU97}. 

Let us define a \emph{run} in a string as a 
maximal substring repeating one symbol. 
The parsing proceeds in two passes. First, it groups the runs into 
\emph{metasymbols}, which are seen as single symbols. 
The resulting sequence is denoted $\hat{S}[1\dd \hat{n}]$. 
The following definition describes the process precisely and defines mappings between $S$ and $\hat{S}$.

\newcommand{\run}[1]{\textrm{\framebox{$#1$}}}
\begin{definition} \label{def:Sb}
The string $\hat{S}[1\dd \hat{n}]$ is obtained from a string $S[1\dd n]$ by replacing every
distinct run $a^\ell$ in $S$ by a special metasymbol $\run{a^\ell}$
so that two occurrences of the same run $a^\ell$ are replaced by the same
metasymbol. The alphabet $\hat{\Sigma}$ of $\hat{S}$ consists of the metasymbols that represent runs in $S$,
that is $\hat{\Sigma} = \{\textrm{\framebox{$a^\ell$}}: a^\ell\text{ is a run in }S\}$. 

A position $S[i]$ that belongs to a run $a^\ell$ is \emph{mapped} to the position $\hat{S}[\hat{i}]$
of the corresponding metasymbol $\run{a^\ell}$, denoted $\hat{i}=map(i)$.
A position $\hat{S}[\hat{i}]$ is \emph{mapped back}
to the maximal range $imap(\hat{i}) = [fimap(\hat{i})\dd limap(\hat{i})]$ of positions in 
$S$ that map to $\hat{i}$. That is, if $S[i\dd i+\ell-1]$ is a run in $S$ that
maps to $\hat{i}$, then $fimap(\hat{i})=i$ and $limap(\hat{i})=i+\ell-1$.
\end{definition}

The string $\hat{S}$ is then parsed into {\em blocks}. A bijective function $\pi : 
\Sigma \rightarrow [1\dd |\Sigma|]$ is chosen uniformly at random; we call it
a {\em permutation}. We then extend $\pi$ to $\hat{\Sigma}$ so that $\pi(\run{a^\ell}) = \pi(a)$,
that is, the value on a metasymbol is inherited from the underlying symbol. Note that no two consecutive symbols in $\hat{S}$ have the same $\pi$ value.
We then define local minima in $\hat{S}$, and these are used to parse $\hat{S}$ (and $S$) into blocks.

\begin{definition} \label{def:localmin}
Given a string $S$, its corresponding string $\hat{S}[1\dd \hat{n}]$, and a permutation 
$\pi$ on the alphabet of $S$, a \emph{local minimum} of $\hat{S}$ is defined as any
position $\hat{i}$ such that $1<\hat{i}<\hat{n}$ and  $\pi(\hat{S}[\hat{i}-1]) > \pi(\hat{S}[\hat{i}]) < 
\pi(\hat{S}[\hat{i}+1])$. 
\end{definition}

\begin{definition} \label{def:parse}
The {\em parsing} of $\hat{S}$ partitions it into a sequence of {\em blocks}. The
blocks end at position $\hat{n}$ and at every local 
minimum. The parsing of $\hat{S}$ induces a parsing on $S$: If a block ends at 
$\hat{S}[\hat{i}]$, then a block ends at $S[limap(\hat{i})]$.
\end{definition}

Note that, by definition, the first block starts at $S[1]$.
When applied on texts $S[1\dd n]$, it will hold that $\hat{S}[1]=\#$ and 
$\hat{S}[\hat{n}]=\$$, so $\hat{S}$ will also be a text. Further, we will always force that
$\pi(\$)=1$ and $\pi(\#)=2$, which guarantees that there cannot be local minima
in $\hat{S}[1\dd 2]$ nor in $\hat{S}[\hat{n}-1\dd \hat{n}]$. Together with the fact that there cannot
be two consecutive local minima, this yields the following observation.

\begin{observation} \label{obs:length2}
Every block in $S$ or $\hat{S}$ is formed by at least two consecutive elements 
(symbols or metasymbols, respectively).
\end{observation}

\begin{definition}\label{def:bb}
    We say that a position $p<n$ of the parsed text $S$ is a \emph{block boundary}
    if a block ends at position $p$.
    For every non-empty fragment $S[i\dd j]$ of $S$, we define \[B(i,j)=\{p - i : i \le p < j\text{ and $p$ is a block boundary}\}.\]
    Moreover, for every integer $c\ge 0
    $, we define subsets $L(i,j,c)$ and $R(i,j,c)$ of $B(i,j)$ consisting of the $\min(c,|B(i,j)|)$ smallest and largest elements of $B(i,j)$, respectively.    
\end{definition}

Observe that any fragment $S[i\dd j]$ intersects a sequence of $1+|B(i,j)|$ blocks
(the first and
the last block might not be contained in the fragment). We are interested in 
\emph{locally contracting} parsings, where this number of blocks
is smaller than the fragment's length by a constant factor.

\begin{definition} \label{def:locally-contractive}
A parsing is \emph{locally contracting} if there exist constants $\alpha$ and 
$\beta<1$ such that $|B(i,j)|\le \alpha+\beta |S[i\dd j]|$ for every fragment $S[i\dd j]$ of $S$.
\end{definition}

\begin{lemma} \label{lem:locally-contracting}
The parsing of $S$ from \Cref{def:parse} is locally contracting with $\alpha=0$ and $\beta=\frac12$.
\end{lemma}
\begin{proof}
By \Cref{obs:length2}, adjacent positions in $S$ cannot both be block boundaries.
Hence, $|B(i,j)|\le \lceil \frac{j-i}{2}\rceil = \lfloor\frac{j-i+1}{2}\rfloor = \lfloor\frac12|S[i\dd j]|\rfloor \le \frac12|S[i\dd j]|$.
\end{proof}
We formally define \emph{locally consistent} parsings as follows.
\begin{definition}\label{def:locally-consistent}
A parsing is \emph{locally consistent} if there exists a constant $c_{p}$ 
such that for every pair of matching fragments $S[i\dd j]=S[i'\dd j']$ it holds that
$B(i,j)\sm B(i',j') \sub L(i,j,c_p)\cup R(i,j,c_p)$, that is $B(i,j)$ and $B(i',j')$
differ by at most $c_p$ smallest and $c_p$ largest elements.
\end{definition}

Next, we prove local consistency of our parsing.
\begin{lemma}\label{lem:alt}
    The parsing of $S$ from \Cref{def:parse} is locally consistent with $c_p=1$.
    More precisely, if $S[i\dd j]=S[i'\dd j']$ are matching fragments of $S$,
    then \[B(i,j)\sm \{limap(map(i))-i\} =  B(i',j') \sm \{limap(map(i))-i\}.\] 
\end{lemma}
\begin{proof}
By definition, a block boundary is a position $q$ such that 
$q=limap(\hat{q})$ for a local minimum $\hat{S}[\hat{q}]$ in $\hat{S}$.
Hence, a position $q$, with $1 < q < n$, is a block boundary if and only if $\pi(S[q])< \pi(S[q+1])$
and $\pi(S[q]) < \pi(S[r])$, where $r=fimap(map(q))-1$ is the rightmost position to left of $q$ with $S[r]\ne S[q]$.

Consider a position $p$, with $i < p < j$, and the corresponding position $p'=p-i+i'$.
If $p > limap(map(i))$, then the positions $r=fimap(map(p))-1$ and $r'=fimap(map(p'))-1$
satisfy $r'-i'=r-i\ge 0$. Hence, $p$ is a block boundary if and only if $p'$ is one.
On the other hand, if $p < limap(map(i))$, then neither $p$ nor $p'$ is a block boundary
because $S[p]=S[p+1]$ and $S[p']=S[p'+1]$. 

Consequently, only the position $p = limap(map(i))$ is a block boundary not necessarily if and only if $p'$ is one.
That is, $B(i,j)\sm  \{limap(map(i))-i\}=B(i',j')\sm  \{limap(map(i))-i\}$.
Moreover, since $limap(map(i))-i$ may only be the leftmost element of $B(i,j)$,
this yields $B(i,j)\sm B(i',j')\sub L(i,j,1)$, and therefore the parsing is locally consistent with $c_p=1$.
\end{proof}

 We conclude this section by defining
\emph{block extensions} and proving that they are sufficiently 
long to ensure that the block is preserved within the occurrences of its extension.
This property will be use several times in subsequent sections.

\begin{definition} \label{def:extended-block}
Let $S[i\dd j]$, with $1<i<j<n$, be a block in $S$. The \emph{extension} of the block $S[i\dd j]$
is defined as $S[i^e\dd j^e]$, where
$i^e=fimap(map(i-1))-1$ and $j^e=j+1$. 
\end{definition}

Note that the first and last blocks cannot be extended. 
For the remaining blocks $S[i\dd j]$, the definition is sound because $map(i-1)>1$ and $j<n$ since $map(i-1)$ and $map(j)$ are local minima of $S'$. Further, note that the block extension spans only
the last symbol of the metasymbol $\hat{S}[map(i^e)]$ and the first of $\hat{S}[map(j^e)]$.

\begin{lemma} \label{lem:extend}
Let $S[i^e\dd j^e]$ be the extension of a block $S[i\dd j]$.
If $S[r'\dd s']$ matches $S[i^e\dd j^e]$, then
$S[r'\dd s']$ contains the same block $S[r\dd s]=S[i\dd j]$, whose extension is
precisely $S[r^e\dd s^e]=S[r'\dd s']$. Furthermore, $r-r^e=i-i^e$ and $s^e-s = j^e-j$.
\end{lemma}
\begin{proof}
Observe that $limap(map(i^e)) = i^e$, so \Cref{lem:alt} yields $B(r',s')\sm \{0\} = B(i^e,j^e)\sm \{0\}$.
Moreover, $map(i^e)=map(i-1)-1$ and $map(j^e)=map(j)+1$,
so $B(i^e,j^e) = \{i-1-i^e, j-i^e\}$ due to \Cref{obs:length2}.
Hence, $ B(r',s') \sm \{ 0 \} = \{i-1-i^e,j-i^e\}$,
and therefore $S[r\dd s]$ is a block, where $r = i-i^e+r'$ and $s = j-i^e+r'$. \Cref{fig:extend} gives an example.

To complete the proof, notice that $s^e=s+1=s'$ and $r^e=limap(map(r-1))-1=r'$ follows from the fact that $S[r'\dd s']$ and $S[i^e\dd j^e]$ match.
\end{proof}

\begin{figure}[t]
\begin{center}
\includegraphics[width=0.4\textwidth]{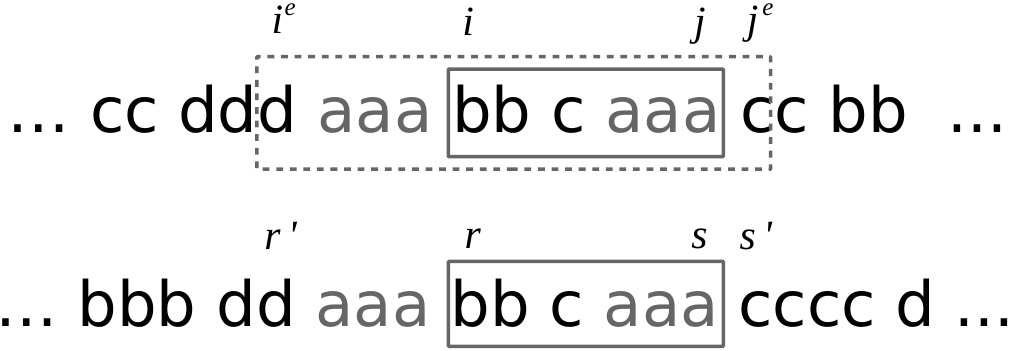}
\end{center}
\caption{Illustration of Lemma~\ref{lem:extend}. Local minima are
shown in gray. Recall $r'=r^e$ and $s'=s^e$.}
\label{fig:extend}
\end{figure}


\section{Grammars with Locality Properties}
\label{sec:locality}

Consider a context-free grammar (CFG) that generates a string $S$ and only $S$ 
\cite{KY00}. Each nonterminal must be the left-hand side in exactly one production, 
and the {\em size} $g$ of the grammar is the sum of the right-hand sides of the productions.
It is NP-complete to compute the smallest grammar for a string $S$
\cite{DBLP:journals/tcs/Rytter03,DBLP:journals/tit/CharikarLLPPSS05}, but it is possible to build grammars of size 
$g = O(z\log(|S|/z))$ if the Lempel--Ziv parsing of $S$ consists of $z$ phrases \cite[Lemma~8]{DBLP:conf/esa/Gawrychowski11}.\footnote{There are older
constructions \cite{DBLP:journals/tcs/Rytter03,DBLP:journals/tit/CharikarLLPPSS05}, but they refer to a restricted 
Lempel--Ziv variant where sources and phrases cannot overlap.}

If we allow, in addition, rules of the form $A \rightarrow A_1^s$, where $s\ge 2$, taken to be
of size 2 for technical convenience, the result is a {\em run-length context-free 
grammar (RLCFG)} \cite{DBLP:conf/mfcs/NishimotoIIBT16}. These grammars encompass CFGs and are 
intrinsically more powerful; for example, the 
smallest CFG for the string family $S=a^n$ is of size $\Theta(\log n)$ whereas already an RLCFG of size $O(1)$ can 
generate it. 

The {\em parse tree} of a CFG has internal nodes labeled with nonterminals
and leaves labeled with terminals. The root is the initial symbol and the 
concatenation of the leaves yields $S$: the $i$th leaf is labeled $S[i]$. 
If $A \rightarrow A_1\cdots A_s$, then any
node labeled $A$ has $s$ children, labeled $A_1, \ldots, A_s$. 
In the parse tree of a RLCFG, rules $A \rightarrow A_1^s$ are represented as a node labeled $A$ with $s$ children nodes labeled $A_1$. The following definition describes the substring of $S$ generated by each node.

\begin{definition}
  If the leaves descending from a parse tree node $v$ are the $i$th to the $j$th leaves, we say that $v$ {\em generates} $S[i\dd j]$ and that $v$ is {\em
  projected} to the interval $proj(v) = [i\dd j]$.
  \end{definition}
  
The subtrees of equally labeled nodes are identical and generate the same strings, so we speak of the strings generated by the grammar symbols. 
We call $\exp(A)$ the {\em expansion} of nonterminal $A$, that is, the string
it generates (or the concatenation of the leaves under any node labeled $A$ in the
parse tree), and $|A| = |\exp(A)|$. For terminals $a$, we assume $\exp(a)=a$.

A grammar is said to be {\em balanced} if
the parse tree is of height $O(\log n)$. A stricter concept is the following one.

\begin{definition} \label{def:locally-balanced}
A grammar is {\em locally balanced} if there exists a constant $b$ such
that, for any nonterminal $A$, the height of any parse tree node labeled $A$ is at most
$b \cdot \log |A|$.
\end{definition}

\subsection{From parsings to balanced grammars}
\label{sec:ec}

We build an RLCFG on a 
text $T[1\dd n]$ using our parsing of Section~\ref{sec:lcp}. 
In the first pass, we collect the distinct runs
$a^\ell$ with $\ell\ge 2$ and create run-length nonterminals of the form $A \rightarrow 
a^\ell$ to replace the corresponding runs in $T$. The resulting sequence is 
analogous to $\hat{T}$, where a nonterminal $A \to a^\ell$ stands for the metasymbol
\framebox{$a^\ell$}, and the terminal $a$ stands for the metasymbol $\run{a^1}$.

Next, we choose a permutation $\pi$ and perform a
pass on the new text $\hat{T}$, defining the blocks based on local minima according to \Cref{def:parse}. Each 
distinct block $A_1\cdots A_k$ is replaced by a distinct nonterminal $A$ with the rule $A \rightarrow A_1 \cdots A_k$ (each $A_i$ can be a symbol of $\Sigma$ 
or a run-length nonterminal created in the first pass). 
The blocks are then replaced by those created nonterminals $A$, which results in a
string $T'$.
The string $T'$ is of length $n' \le \lfloor n/2 \rfloor$, by 
Observation~\ref{obs:length2}. Note that the first and last symbols of
$T'$ expand to blocks that contain \# and \$, respectively, and thus they are 
unique too. We can
then regard $T'$ as a text, by having its first nonterminal, $T'[1]$,
play the role of \#, and the last, $T'[n']$, play the role of \$. 

The process is then repeated again on $T'$, and iterated for $h\le \lfloor\log n\rfloor$ rounds, until a single nonterminal is obtained. This is 
the initial symbol of the grammar. We denote by $T_r[1\dd n_r]$ the text created in
round $r$, so $T_0=T$ and $T_1=T'$. We also denote by $\hat{T}_r[1\dd \hat{n}_r]$ the intermediate
text obtained by collapsing runs in $T_r$.
Figure~\ref{fig:grammar} exemplifies the grammars we build and the corresponding parse tree.

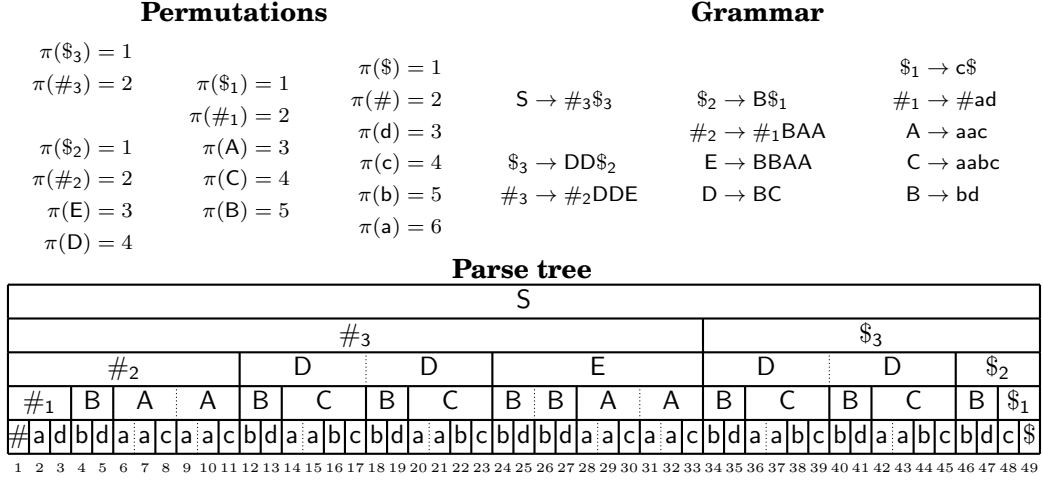
\begin{figure}[ht]
  \begin{center}
    \begin{minipage}{.45\textwidth}
       \centering{\bf Permutations}
        \end{minipage}%
        \begin{minipage}{.55\textwidth}
            \centering{\bf Grammar}
             \end{minipage}
    \begin{minipage}{.15\textwidth}
    \footnotesize
        \begin{align*}
            \pi(\mathsf{\$_3}) &=  1\\
            \pi(\mathsf{\#_3}) &=  2\\
            & \\
            \pi(\mathsf{\$_2}) &=  1\\
            \pi(\mathsf{\#_2}) &=  2\\
            \pi(\mathsf{E}) &=  3\\
            \pi(\mathsf{D}) &=  4
        \end{align*}
    \end{minipage}%
    \begin{minipage}{.15\textwidth}
        \footnotesize
            \begin{align*}
                \pi(\mathsf{\$_1}) &= 1\\
                \pi(\mathsf{\#_1}) &= 2\\
                \pi(\mathsf{A}) &= 3\\
                \pi(\mathsf{C}) &= 4\\
                \pi(\mathsf{B}) &= 5
            \end{align*}
        \end{minipage}%
        \begin{minipage}{.15\textwidth}
            \footnotesize
                \begin{align*}
                    \pi(\mathsf{\$}) &= 1\\
                    \pi(\mathsf{\#}) &= 2\\
                    \pi(\mathsf{d}) &= 3\\
                    \pi(\mathsf{c}) &= 4\\
                    \pi(\mathsf{b}) &= 5\\
                    \pi(\mathsf{a}) &= 6
                \end{align*}
            \end{minipage}%
    \begin{minipage}{.18\textwidth}
        \footnotesize
        \begin{align*}
            \mathsf{S} &\to \mathsf{\#_3\$_3}\\
            & \\
            \mathsf{\$_3} &\to \mathsf{DD\$_2}\\
            \mathsf{\#_3} &\to \mathsf{\#_2 DDE}
        \end{align*}
    \end{minipage}%
    \begin{minipage}{.18\textwidth}
        \footnotesize
        \begin{align*}
            \mathsf{\$_2} &\to \mathsf{B\$_1}\\
            \mathsf{\#_2} &\to \mathsf{\#_1 B AA}\\
            \mathsf{E} &\to \mathsf{BBAA}\\
            \mathsf{D} &\to \mathsf{BC}
        \end{align*}
    \end{minipage}%
    \begin{minipage}{.18\textwidth}
        \footnotesize
        \begin{align*}
            \mathsf{\$_1} &\to \mathsf{c\$}\\
            \mathsf{\#_1} &\to \mathsf{\#ad}\\
            \mathsf{A} &\to \mathsf{aac}\\
            \mathsf{C} &\to \mathsf{aabc}\\
            \mathsf{B} &\to \mathsf{bd}\\
        \end{align*}
    \end{minipage}

{\bf Parse tree}

    \begin{tikzpicture}[xscale=0.56, yscale=0.9]
    \foreach \x[count=\i] in {\#,a,d,b,d,a,a,c,a,a,c,b,d,a,a,b,c,b,d,a,a,b,c,b,d,b,d,a,a,c,a,a,c,b,d,a,a,b,c,b,d,a,a,b,c,b,d,c,\$}{
        \draw (\i/2, 0) node[anchor=mid] {$\mathsf{\x}$};
    }
    \foreach \i in {0,1,...,49}{
        \draw[densely dotted] (\i/2+.25, -.25) -- (\i/2+.25, .25);
    }
    \foreach \i in {0,1,2,3,4,5,7,8,10,11,12,13,15,16,17,18,19,21,22,23,24,25,26,27,29,30,32,33,34,35,37,38,39,40,41,43,44,45,46,47,48,49}{
        \draw[thick] (\i/2+.25, -.25) -- (\i/2+.25, .25);
    }
    \foreach \x/\i in {\#_1/2, B/4.5, A/7, A/10, B/12.5, C/15.5, B/18.5, C/21.5, B/24.5, B/26.5, A/29, A/32, B/34.5, C/37.5, B/40.5, C/43.5, B/46.5, \$_1/48.5}  {
        \draw (\i/2, .5) node[anchor=mid] {$\mathsf{\x}$};
    }
    \foreach \i in {0,3,5,8,11,13,17,19,23,25,27,30,33,35,39,41,45,47,49}{
        \draw[densely dotted] (\i/2+.25, .25) -- (\i/2+.25, .75);
    }
    \foreach \i in {0,3,5,11,13,17,19,23,27,33,35,39,41,45,47,49}{
        \draw[thick] (\i/2+.25, .25) -- (\i/2+.25, .75);
    }
    \foreach \x/\i in {\#_2/6, D/14.5, D/20.5, E/28.5, D/36.5, D/42.5, \$_2/47.5}  {
        \draw (\i/2, 1) node[anchor=mid] {$\mathsf{\x}$};
    }
    \foreach \i in {0,11,17,23,33,39,45,49}{
        \draw[densely dotted] (\i/2+.25, .75) -- (\i/2+.25, 1.25);
    }
    \foreach \i in {0,11,23,33,45,49}{
        \draw[thick] (\i/2+.25, .75) -- (\i/2+.25, 1.25);
    }
    \foreach \x/\i in {\#_3/17, \$_3/41.5}  {
        \draw (\i/2, 1.5) node[anchor=mid] {$\mathsf{\x}$};
    }
    \foreach \i in {0,33,49}{
        \draw[thick] (\i/2+.25, 1.25) -- (\i/2+.25, 1.75);
    }
    \foreach \x/\i in {S/25}  {
        \draw (\i/2, 2) node[anchor=mid] {$\mathsf{\x}$};
    }
    \foreach \i in {0,49}{
        \draw[thick] (\i/2+.25, 1.75) -- (\i/2+.25, 2.25);
    }
    \foreach \y in {-.25,.25,...,2.25}{
        \draw[thick] (0.25, \y) -- (24.75, \y);
    }

    \foreach \i in {1,2,...,49}{
        \draw(\i/2, -.25) node[below]{\tiny $\i$};
    }
\end{tikzpicture}
\end{center}
\caption{An example of the construction of our grammar. 
The top-left part shows the permutations $\pi$ 
assigned in each level, and the top-right part gives the complete grammar built 
(for simplicity we omit run-length nonterminals). The parse tree, shown on the bottom, also omits run-length nonterminals. The texts $T_r$ correspond to the subsequent levels of the parse tree (starting from the bottom).
Level-$r$ block boundaries that are not run boundaries are depicted using dotted lines.
For example, $T_2 = \mathsf{\#_2DDEDD\$_2}$, $n_2=7$, and the level-2 block boundaries are $11,17,23,33,39,45$.
On the other hand, $\hat{n}_2=5$ and the level-2 run boundaries are
$11,23,33,45$. The corresponding parsings $\Parsing_2$ and $\hat{\Parsing}_2$ decompose $T$ as
$\mathsf{\#adbdaacaac| bdaabc | bdaabc | bdbdaacaac |  bdaabc | bdaabc | bdc\$}$ and 
$\mathsf{\#adbdaacaac| bdaabcbdaabc | bdbdaacaac |  bdaabcbdaabc | bdc\$}$, respectively.
}
\label{fig:grammar}
\end{figure}

The height of the grammar is at most 
$2h\le 2\lfloor\log n\rfloor$, because we create run-length rules and then 
block-rules in each round. This grammar is then balanced because, by
Observation~\ref{obs:length2}, $n_r \le n/2^r$. Moreover,
the grammar is locally balanced.

\begin{lemma} \label{lem:locally-balanced}
The grammar we build from our parsing is locally balanced with $b=2$.
\end{lemma}
\begin{proof}
Because of Observation~\ref{obs:length2}, any subtree rooted at a nonterminal 
$A$ in the parse tree (at least) doubles the number of nodes per round towards
the leaves. If $A$ is formed in round $r$, then the subtree has height at
most $2r$, and the expansion satisfies $|A| \ge 2^r$. The height of the subtree
rooted at $A$ is thus at most $2r \le 2 \log |A|$.
\end{proof}

\subsection{Local consistency properties}
We now formalize the concept of \emph{local consistency} for our grammars.
For each $r\in[0\dd h]$, the subsequent characters of $T_r$
naturally correspond to nodes of the parse tree of $T$, and the fragments
$T[i\dd j]$ generated by these nodes form a decomposition of $T$.
We denote this parsing of $T$ by $\Parsing_r$. In other words,
$T[i\dd j]$ is a block of $\Parsing_r$ if and only if $[i\dd j]=proj(v)$ for some node $v$
labeled by a symbol in $T_r$.
We refer to the blocks and block boundaries in this parsing as \emph{level-$r$ blocks}
and \emph{level-$r$ block boundaries}. 
Analogously, we define a parsing $\hat{\Parsing}_r$ with blocks corresponding to subsequent symbols of $\hat{T}_r$,
and we refer to the underlying blocks and block boundaries as \emph{level-$r$ runs} and \emph{level-$r$ run boundaries};
see \Cref{fig:grammar}.

Note that every level-$r$ run boundary is also a level-$r$ block boundary,
and every level-$(r+1)$ block boundary is also a level-$r$ run boundary.
Moreover, by \Cref{obs:length2} at most one out of every two subsequent level-$r$ run boundaries
can be a level-$(r+1)$ block boundary.

\begin{definition}
  For every non-empty fragment $T[i\dd j]$ of $T$,
the sets defined according to \Cref{def:bb} for the parsing $\Parsing_r$ are denoted $B_r(i,j)$, $L_r(i,j,c)$, and $R_r(i,j,c)$. Analogously, we denote by $\hat{B}_r(i,j)$, $\hat{L}_r(i,j,c)$, and $\hat{R}_r(i,j,c)$
the sets defined for the parsing $\hat{\Parsing}_r$.
\end{definition}

These notions let us reformulate \Cref{lem:alt} so that it is directly applicable at every level $r$.
\begin{lemma}\label{lem:altr}
  If matching fragments $T[i\dd j]$ and $T[i'\dd j']$ both consist of full level-$r$ blocks,
  then the corresponding fragments of $T_r$ also match, so
  $B_r(i,j)=B_r(i',j')$ and $\hat{B}_r(i,j)=\hat{B}_r(i',j')$. Moreover, $B_{r+1}(i,j)\sm \{\min \hat{B}_r(i,j)\} = B_{r+1}(i',j')\sm \{\min \hat{B}_r(i,j)\}$  if $\hat{B}_r(i,j)\ne \emptyset$,
  and $B_{r+1}(i,j)=B_{r+1}(i',j')=\emptyset$ otherwise.
\end{lemma}
\begin{proof}
We proceed by induction on $r$. The first two claims hold trivially for $r=0$: the fragments $T[i\dd j]$ and $T[i'\dd j']$ of $T_0=T$ clearly match, and $B_0(i,j)=[0\dd j-i-1]=B_0(i',j')$.
For $r>0$, on the other hand, $T[i\dd j]$ and $T[i'\dd j']$ consist of full level-$(r-1)$ blocks,
so the inductive assumption yields that the corresponding fragments of $T_{r-1}$ also match and
that $B_r(i,j)=B_{r}(i',j')=\emptyset$ or $B_{r}(i,j)\sm \{\min \hat{B}_{r-1}(i,j)\} = B_{r}(i',j')\sm \{\min \hat{B}_{r-1}(i,j)\}$. In the latter case, we observe that $i-1$ and $i+\min \hat{B}_{r-1}(i,j)$ are subsequent level-$(r-1)$ run boundaries while $i-1$ is a level-$r$ block boundary, or $i=1$ and $i+\min \hat{B}_{r-1}(i,j)$ is the leftmost level-$(r-1)$ run boundary. Either way, $i+\min \hat{B}_{r-1}(i,j)$ cannot be a level-$r$ block boundary due to \Cref{obs:length2}, so $B_{r}(i,j)\sm \{\min \hat{B}_{r-1}(i,j)\}=B_r(i,j)$. A symmetric argument proves that $B_{r}(i',j')\sm \{\min \hat{B}_{r-1}(i,j)\}=B_{r}(i',j')$,
which lets conclude that $B_r(i,j)=B_{r}(i',j')$. Hence, the matching fragments of $T_{r-1}$ corresponding to $T[i\dd j]$ and $T[i'\dd j']$
are parsed into the same blocks so the corresponding fragments of $T_r$ also match.

To prove the other two claims for arbitrary $r\ge 0$, notice that the fragments of $T_r$ corresponding to $T[i\dd j]$ and $T[i'\dd j']$  are occurrences of the same string, denoted $P_r$. Hence, $\hat{B}_r(i,j)$ and $\hat{B}_r(i',j')$ are equal as they both correspond to the run boundaries in $P_r$.
If $P_r$ consists of a single run (i.e., if $\hat{B}_r(i,j)=\emptyset$), then clearly $B_{r+1}(i,j)=B_{r+1}(i',j')= \emptyset$.
Otherwise, \Cref{lem:alt} implies  $B_{r+1}(i,j)\sm \{\min \hat{B}_r(i,j)\} = B_{r+1}(i',j')\sm \{\min \hat{B}_r(i,j)\}$.
\end{proof}

Nevertheless, we define \emph{local consistency} of a grammar as a stronger property than the one expressed in \Cref{lem:altr}:
we require that $B_r(i,j)$ and $B_r(i',j')$ resemble each other even if the matching fragments $T[i\dd j]$ and $T[i'\dd j']$ do not consist of full blocks.

\begin{definition} \label{def:lcg}
The grammar we build is {\em locally consistent} if there is a
constant $c_{g}$ such that the parsings $\Parsing_r$ are all locally consistent
with constant $c_g$.
\end{definition}

In the rest of this section, we prove that our grammar is locally consistent
with constant $c_g = 3$. Our main tool is the following construction of sets $B_r(P)$ and $\hat{B}_r(P)$, consisting of the positions (relative to $P$) of \emph{context-insensitive} level-$r$ block and run boundaries
that are common to all occurrences of $P$ in $T$. Despite these sets are defined based on an occurrence of $P$ in $T$,
we show in \Cref{lem:lcg} that they do not depend on the choice of the occurrence.
\begin{definition}\label{def:bp}
  Let $P$ be a substring of $T$ and let $T[i\dd j]$ be its arbitrary occurrence in $T$.
  The sets $B_r(P)$ and $\hat{B}_r(P)$ for $r\ge 0$ are defined recursively, with $X+\delta= \{x+\delta : x\in X\}$.
  \begin{equation*}
    B_r(P)=\left\{\!\begin{aligned}
     &[0\dd |P|-2] && \text{if }r=0,\\
      &B_r(i+1+\min \hat{B}_{r-1}(P),i+\max B_{r-1}(P))+1+\min\hat{B}_{r-1}(P) && \text{if }\hat{B}_{r-1}(P)\ne \emptyset,\\
      &\emptyset && \text{if }\hat{B}_{r-1}(P)=\emptyset;
    \end{aligned}\right.
  \end{equation*}
  \begin{equation*}
     \hat{B}_r(P)=\left\{\begin{aligned}
       &\hat{B}_r(i+1+\min B_r(P), i+\max B_r(P))+1+\min B_r(P) &
       & \text{if }B_r(P)\ne \emptyset,\\
       &\emptyset && \text{if }B_r(P)=\emptyset.
     \end{aligned}\right.
   \end{equation*}
\end{definition}
Our index also relies on aset $M(P)$ designed as a superset of $B_r(i,j)\sm B_r(P)$ for every~$r$ and every occurrence $T[i\dd j]$ of $P$.
In other words, $M(P)$ contains, for each $r$, positions within $P$ that may be level-$r$ block boundaries in some
but not necessarily all occurrences of~$P$. 
\begin{definition}\label{def:m}
  For a substring $P$ of $T$, the set $M(P)$ is defined to contain $\min B_r(P)$ and $\max B_r(P)$ for every $r\ge 0$ with $B_r(P)\ne \emptyset$,
  and $\min \hat{B}_r(P)$ for every $r\ge 0$ with $\hat{B}_r(P)\ne \emptyset$. 
\end{definition}

\begin{example}
  Consider $P= \mathsf{dbdaacaacbdaabcbdaabcbd}$
  with occurrences $T[3\dd 25]$ and $T[25\dd 47]$ in text of \Cref{fig:grammar}.
  For $r=0$, we define $B_0(P)=[0\dd 21]$ and set $\hat{B}_0(P)=\{1,2,4,5,7,8,9,10,12,13,14,15,16,18,19,20\}=1+\hat{B}_0(4,24)=1+\hat{B}_0(26,46)$.
  For $r=1$, we set $B_1(P)=\{2,5,8,10,14,16,20\}=2+B_1(5,24)=2+B_1(27,46)$
  and $\hat{B}_1(P) = \{8,10,14,16\}=3+\hat{B}_1(6,23)=3+\hat{B}_1(28,45)$.
  For $r=2$, we set $B_2(P)=\{14\}=9+B_2(12,19)=9+B_2(34,41)$
  and $\hat{B}_2(P)=\emptyset=15+\hat{B}_2(18,17)=15+\hat{B}_2(40,39)$.
  For $r\ge 3$, we have $B_r(P)=\hat{B}_r(P)=\emptyset$.
  Consequently, $M(P)=\{0,1,2,8,14,20,21\}$.
  \end{example}

We now show $B_r(P)$ contains all the level-$r$ block boundaries in any occurrence of $P$ in $T$ except possibly the first 3 and the last one, but those missing boundaries belong to $M(P)$.

\begin{lemma}\label{lem:lcg}
  For every substring $P$ of $T$ and every $r\ge 0$,
  the sets $B_r(P)$ and $\hat{B}_r(P)$ do not depend on the choice of an occurrence $T[i\dd j]$ of $P$.
  Moreover,
  \begin{equation} \label{eq:lcg}
  B_r(P)\cup L_r(i,j,3)\cup R_r(i,j,1)=B_r(i,j)\sub B_r(P)\cup M(P).
  \end{equation}
\end{lemma}
\begin{proof}
We proceed by induction on $r$, proving the independence of $\hat{B}_r(P)$ only at step $r+1$.
In the base case, $B_0(P)=[0\dd |P|-2]$ 
does not depend on the choice of the occurrence, and
Eq.~\eqref{eq:lcg} is satisfied because $B_0(P)=B_0(i,j)$.

For the inductive step, we assume the claims hold for $B_r(P)$.
If $B_r(P)=\emptyset$, then $\hat{B}_r(P)=B_{r+1}(P)=\emptyset$ do not depend on the occurrence of $P$.
The inductive assumption yields $B_{r+1}(i,j)\sub B_r(i,j)\sub M(P) = B_{r+1}(P) \cup M(P)$ and $|B_{r+1}(i,j)| \le |B_r(i,j)|=|L_r(i,j,3)\cup R_r(i,j,1)|\le 4$,
so $L_{r+1}(i,j,3)\cup R_{r+1}(i,j,1)=B_{r+1}(i,j)$ and Eq.~\eqref{eq:lcg} is satisfied.

We henceforth assume that $B_r(P)\ne \emptyset$.
Since $B_r(P)\sub B_r(i,j)$, both $i+\min B_r(P)$ and $i+\max B_r(P)$ are level-$r$ block boundaries,
and therefore $T[i+\min B_r(P)+1\dd i+\max B_r(P)]$ consists of full level-$r$ blocks.
We conclude from \Cref{lem:altr} that $\hat{B}_r(P)$, as defined in \Cref{def:bp}, does not depend on the occurrence of $P$.
Moreover, the only position between $i+\min B_r(P)$ and $i+\max B_r(P)$ that may or may not be a level-$(r+1)$ block boundary depending on the context of $T[i\dd j]$
is $i+\min \hat{B}_r(P)$ provided that $\hat{B}_r(P) \ne \emptyset$. In particular, $B_{r+1}(P)$, as defined in \Cref{def:bp},
also does not depend on the occurrence~of~$P$.

To prove that $B_{r+1}(P)$ satisfies Eq.~\eqref{eq:lcg}, we consider two cases. First, suppose that $\hat{B}_r(P)=\emptyset$, that is,
there are no level-$r$ run boundaries between $i+\min B_r(P)$ and $i+\max B_r(P)$. 
Since $\hat{B}_r(i,j)\sub B_r(i,j)$, the inductive assumption $B_r(i,j)=B_r(P)\cup L_r(i,j,3)\cup R_r(i,j,1)$ implies 
$\hat{B}_{r}(i,j)\sub  \{\min B_r(P), \max B_r(P)\} \cup L_{r}(i,j,3)\cup R_r(i,j,1)$,
while $B_r(i,j)\sub B_r(P) \cup M(P)$ yields $\hat{B}_{r}(i,j) \sub  \{\min B_r(P), \max B_r(P)\} \cup M(P) = M(P)$,
where the equality follows from \Cref{def:m}.
The former assertions yields $|\hat{B}_r(i,j)|\le 6$,
and since $B_{r+1}(i,j)\sub \hat{B}_r(i,j)$ cannot contain two consecutive elements of $\hat{B}_r(i,j)$ by \Cref{obs:length2}, 
we conclude that $|B_{r+1}(i,j)|\le 3$.
In particular, since $B_{r+1}(P)=\emptyset$ according to \Cref{def:bp}, we have $B_{r+1}(P) \cup L_{r+1}(i,j,3)\cup R_{r+1}(i,j,1)=B_{r+1}(i,j)\sub B_{r+1}(P) \cup M(P)$ as claimed.

Next, suppose that $\hat{B}_r(P)\ne\emptyset$. \Cref{def:bp} clearly implies $B_{r+1}(P)\sub B_{r+1}(i,j)$,
so it remains to prove that $B_{r+1}(i,j)$ is a subset of both $B_{r+1}(P)\cup L_{r+1}(i,j,3)\cup R_{r+1}(i,j,1)$ and $B_{r+1}(P)\cup M(P)$.
We take $q\in B_{r+1}(i,j)$ and consider three cases. 
\begin{enumerate}[(1)]
\item If $q\le \min \hat{B}_r(P)$, then $q\in (L_r(i,j,3)\cap M(P))\cup \{\min B_r(P),\min \hat{B}_r(P)\}$ and therefore $q\in \hat{L}_r(i,j,5)$.%
\footnote{By the choice of $\hat{B}_r(P)$ in \Cref{def:bp}, there are no level-$r$ run boundaries between $\min B_r(P)$ and $\min \hat{B}_r(P)$.
  Note that $q < \min B_r(P)$ yields $q \not\in B_r(P)$. Since $q\in B_r(i,j)$, by the inductive assumption $q\notin B_r(P)$ implies $q\in L_r(i,j,3)\cap M(P)$ ($q\notin R_r(i,j,1)$ because $q < \min B_r(P)\in B_r(i,j)$).  For the same reason, $L_r(i,j,3)\cup \{\min B_r(P)\}\sub L_r(i,j,4)$
  and $\min \hat{B}_r(P)\in \hat{L}_r(i,j,5)$.}
Since $B_{r+1}(i,j)$ cannot contain two consecutive elements of $\hat{B}_r(i,j)$ due to \Cref{obs:length2},
$q\in B_{r+1}(i,j) \cap \hat{L}_r(i,j,5)$ implies $q\in L_{r+1}(i,j,3)$. 
Finally, $q\in M(P)\cup \{\min B_r(P),\min \hat{B}_r(P)\} = M(P)$, where the equality holds due to \Cref{def:m}.
\item If $q \ge \max B_r(P)$, then $q\in (R_r(i,j,1)\cap M(P))\cup \{\max B_r(P)\}$ and therefore $q\in \hat{R}_r(i,j,2)$.%
\footnote{Note that $q > \max B_r(P)$ yields $q \not\in B_r(P)$. Since $q\in B_r(i,j)$, by the inductive assumption $q\notin B_r(P)$ implies $q\in R_r(i,j,1)\cap M(P)$ ($q\notin L_r(i,j,3)$ because $q > \max B_r(P) > \min \hat{B}_r(P) > \min B_r(P)$ and these 3 elements belong to $B_r(i,j)$). For the same reason,  $R_r(i,j,1)\cup \{\max B_r(P)\} \sub R_r(i,j,2)$.}
Since $B_{r+1}(i,j)$ cannot contain two consecutive elements of $\hat{B}_r(i,j)$ due to \Cref{obs:length2},
$q\in B_{r+1}(i,j) \cap R_r(i,j,2)$ implies $q\in R_{r+1}(i,j,1)$. Finally, $q\in M(P)\cup \{\max B_r(P)\} = M(P)$, where the equality holds due to \Cref{def:m}.
\item If $\min \hat{B}_r(P) < q < \max B_r(P)$, then $q+i$ is a level-$(r+1)$ block boundary and $q\in B_{r+1}(P)$ by \Cref{def:bp}.\qed
\end{enumerate}
\end{proof}

  \Cref{lem:lcg} implies that the grammar constructed in this section is locally consistent with $c_g=3$. 
  We conclude this section with a further characterization of the set $M(P)$.

  \begin{lemma}\label{lem:m}
    For each substring $P=T[i\dd j]$, the set $M(P)$ satisfies the following properties:
    \begin{enumerate}[(a)]
      \item\label{it:mleft} If $B_r(i,j)\ne \emptyset$ for some $r\ge 0$, then $\min B_r(i,j)\in M(P)$,
      \item\label{it:mleftprim} If $\hat{B}_r(i,j)\ne \emptyset$ for some $r\ge 0$, then $\min \hat{B}_r(i,j)\in M(P)$,
      \item\label{it:msize} $|M(P)|\le 3\lceil\log |P|\rceil$.
    \end{enumerate}
  \end{lemma}
  \begin{proof} 
    To prove \eqref{it:mleft} for any $r$,
    note that $B_r(P) \sub B_{r}(i,j)\sub B_r(P)\cup M(P)$ by \Cref{lem:lcg}.
    If $\min B_r(i,j) \notin B_r(P)$, then it belongs to $M(P)$.
    Otherwise, it must be equal to $\min B_r(P)$, which is in $M(P)$ by \Cref{def:m}.

    The proof of \eqref{it:mleftprim} is similar:
    Since $\hat{B}_r(i,j) \sub B_r(i,j)$, either $\min \hat{B}_r(i,j)\in B_r(i,j) \sm B_r(P) \sub M(P)$,
    or $\min \hat{B}_r(i,j)\in B_r(P)$. If $\min \hat{B}_r(i,j) \in \{\min B_r(P),\max B_r(P)\}$, then it is in $M(P)$ by \Cref{def:m}.
    Otherwise, $\min B_r(P) < \min \hat{B}_r(i,j) < \max B_r(P)$
    and, by the choice of $\hat{B}_r(P)$ in \Cref{def:bp},  $\min \hat{B}_r(i,j)=\min \hat{B}_r(P)$ is also in $M(P)$ by \Cref{def:m}.

    To prove \eqref{it:msize}, notice that $|B_r(P)| \le \frac12|B_{r-1}(P)|$ holds for all $r$
 due to \Cref{def:bp}, \Cref{obs:length2}, and $\min B_{r-1}(P)\notin B_r(P)$. 
    This implies
    $|B_r(P)| \le |B_0(P)| \cdot 2^{-r} < |P|\cdot 2^{-r}$,
    and therefore $\hat{B}_r(P) = B_r(P) = \emptyset$ for $r \ge \log |P|$.
    \Cref{def:m} now yields the claim.
  \end{proof}


\section{Bounding our Grammar in terms of Attractors} \label{sec:attractor}

Let us first define the concept of attractors in a string \cite{KP18}.

\begin{definition}[\cite{KP18}] \label{def:attractor}
An \emph{attractor} of a string $S$ a set $\Gamma\sub[1\dd n]$ of positions in $S$
such that each non-empty substring 
$Q$ of $S$ has an occurrence $S[i\dd j]$ containing an attractor
position, i.e., satisfying $i\le p \le j$ for some $p\in \Gamma$.
\end{definition}

In this section, we show that the RLCFG of Section~\ref{sec:locality} is of size 
$g = O(\gamma\log(n/\gamma))$, where $\gamma$ is the minimum size of an 
attractor of $T$. The key is to prove that distinct nonterminals are 
formed only around the attractor elements. 
For this, we first prove that
$T'[1\dd n']$, where the blocks of $T$ are converted into nonterminals,
contains an attractor of size at most $3\gamma$. 

\begin{lemma} \label{lem:4gamma}
Let $\Gamma$ be an attractor of $T$, and let $\Gamma' = \bigcup_{p\in \Gamma} [p'-1\dd p'+1]$, where $p'$ is the 
position in $T'$ of the nonterminal that covers $p$ in $T$. 
Then $\Gamma'$ is an attractor of $T'$.
\end{lemma}

\begin{figure}[t]
    \begin{center}
    \includegraphics[width=\textwidth]{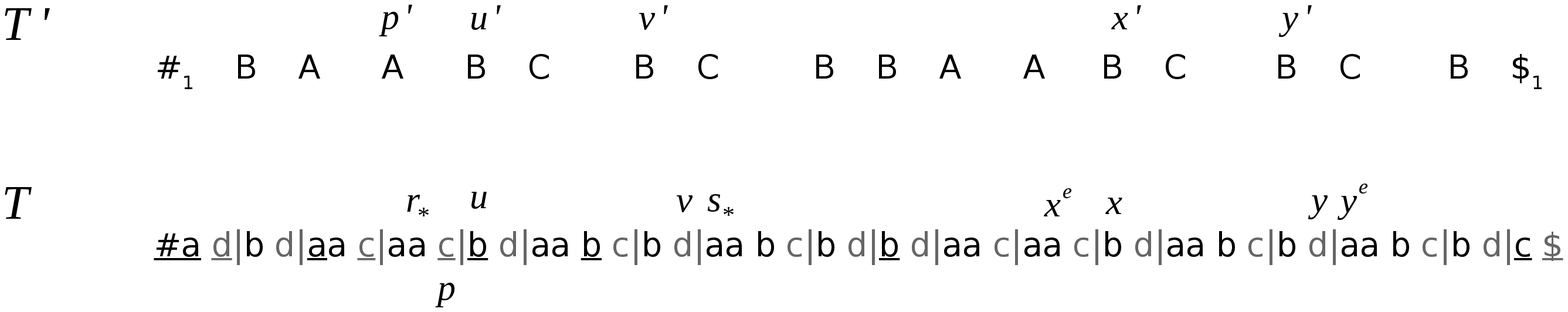}
    \end{center}
    \caption{Illustration of Lemma~\ref{lem:4gamma}. We underlined positions in $T$ corresponding to a particular string attractor.}
    \label{fig:4gamma}
    \end{figure}

\begin{proof}
Figure~\ref{fig:4gamma} illustrates the proof.
Consider an arbitrary substring $T'[x'\dd y']$, with $x' \ge 3$ and $y' \le
n'-2$; otherwise the substring crosses an attractor because $1$ and $n$ are
in $\Gamma$. This is a sequence of consecutive
nonterminals, each corresponding to a block in $T$. Let $T[x\dd y]$
be the substring of $T$ formed by all the blocks that map to $T'[x'\dd y']$. The
union of their extensions is also a substring $T[x^e\dd y^e]$ of $T$. Since $\Gamma$ is an attractor in $T$, 
there exists a copy $T[r_*\dd s_*] = T[x^e\dd y^e]$ that includes an element $p\in\Gamma$, 
$r_* \le p \le s_*$. 

Consider any block $T[i\dd j]$ inside $T[x\dd y]$. Its extension
$T[i^e\dd j^e]$ is contained in $T[x^e\dd y^e]$,
so a copy 
$T[r'\dd s']$ of $T[i^e\dd j^e]$ appears inside $T[r_*\dd s_*]$. By 
Lemma~\ref{lem:extend}, the block $T[i\dd j]$ also forms a 
block $T[r\dd s]$ inside $T[r'\dd s']$, at the same relative position; furthermore, 
$T[r'\dd s']=T[r^e\dd s^e]$ is the extension of $T[r\dd s]$.

Since this happens for every block $T[i\dd j]$ inside $T[x\dd y]$,
which is a sequence of blocks, it follows that $T[x\dd y]$ 
appears inside $T[r_*\dd s_*]$, as a subsequence $T[u\dd v]$ of blocks;
furthermore, its extension $T[u^e\dd v^e]$ coincides with $T[r_*\dd s_*]$ and thus 
contains $p$.
Moreover, $T[u\dd v]$ maps to a substring $T'[u'\dd v']=T'[x'\dd y']$.

Since $v^e = v+1$ and $u^e = fimap(map(u-1))-1$, due to Observation~\ref{obs:length2},
the fragments $T[u^e \dd u-1]$ and $T[v+1\dd v^e]$ are contained within single blocks. Therefore, 
the position $p'$ to which 
$p$ is mapped in $T'$ belongs to $T'[u'-1\dd v'+1]$. Consequently, $T'[u'\dd v']$ contains a position in 
$\Gamma'$.
\end{proof}

We now show that the first round contributes $O(\gamma)$ to the size of the
final RLCFG. In this bound, we only count the the sizes of the generated rules;
the whole accounting will be done in Theorem \ref{thm:rlcfg}. The idea
is to show that the $3$ distinct blocks formed around each attractor element
have expected length $O(1)$.

\begin{lemma} \label{lem:linear}
The first round of parsing contributes $O(\gamma)$ to the grammar size, in
expectation. Further, a parsing producing a grammar of size $O(\gamma)$ is
found in $O(n)$ expected time provided that $\gamma$ is known.
\end{lemma}
\begin{proof}
Let us first focus on block-forming rules; we consider the run-length rules in the next paragraph. The right-hand sides of the block-forming rules correspond to the distinct blocks formed in $\hat{T}$, that is, to single symbols in $T'$.
All the distinct symbols in $T'$, in turn, appear at positions of $\Gamma'$.
By \Cref{lem:4gamma}, $\Gamma'$ is of size at most $3\gamma$; therefore, there are at most $3\gamma$ distinct blocks in $\hat{T}$ and in $T$ (i.e., those containing attractor elements of $T$ and their neighboring blocks), and thus at most $3\gamma$ distinct nonterminals are formed in the grammar.

We must also show, however, that the sum of the sizes of the right-hand sides of those $3\gamma$ productions also add up to $O(\gamma)$. Consider a block of $\hat{T}$ of length $\ell$. The right-hand side of its corresponding production is $\ell$. Each element of $\hat{T}$ can be a 
metasymbol, however, so the grammar may indeed include $\ell$ further run-length nonterminals,
contributing up to $2\ell$ to the grammar size. Therefore, each distinct block of length $\ell$ in $\hat{T}$ contributes at most $3\ell$ 
to the grammar size. We now show that $\ell=O(1)$ in expectation for the $3\gamma$ blocks specified above.

Consider an attractor element $p$ and its position $\hat{p}=map(p)$ when mapped to $\hat{T}$.
Let $T[i\dd j]$ be the block containing $p$ and let $T[i'..j']$ be its concatenation
with the adjacent blocks ($T[i'\dd i-1]$ and $T[j+1\dd j']$).
Moreover, let $\hat{T}[\hat{i}\dd \hat{j}]$ and $\hat{T}[\hat{i}'\dd \hat{j}']$ be the corresponding fragments of $\hat{T}$,
with $\hat{i} = map(i)$, $\hat{j} = map(j)$, $\hat{i}' = map(i')$, and $\hat{j}' = map(j)$.

Let $\ell^+ = \hat{j}'-\hat{p}+1$, $\ell^- = \hat{p}-\hat{i}'$, and $\ell = \ell^+ +\ell^-$. Then, $3\ell$
is the maximum possible contribution of attractor $p$ to the grammar size via 
nonterminals that represent these blocks. 

The area $\hat{T}[\hat{p}\dd \hat{j}']$ contains at most 2 local minima,
at $\hat{j}$ and $\hat{j}'$ (unless $\hat{j}' = \hat{n}$).
Note that, between two consecutive local minima, we have a sequence of 
nondecreasing values of $\pi$ and then a sequence of nonincreasing values of 
$\pi$. Our area can be covered by 2 such ranges. Hence, if we split the substring of length $\ell^+$ into 4 equal parts of length
$\ell^+/4$, at least one of them must be monotone (i.e., nondecreasing or 
nonincreasing) with respect to $\pi$.

Note that consecutive symbols in $\hat{T}$ are always different. Further, if there 
are repeated symbols in a length-$d$ substring of $\hat{T}$, then it cannot be monotone with respect to $\pi$. If
all the symbols are different, instead, exactly one out of $d!$ 
permutations $\pi$ will make the substring increasing and one out of $d!$
will make it decreasing, where $d$ is the length of the substring.

As a result, at most 2 out of $(\ell^+/4)!$ permutations can make one of our
length-$(\ell^+/4)$ substrings monotone. If we choose permutations $\pi$ 
uniformly at random, then the probability that at least one of our 4 substrings 
is monotone is at most $8 / (\ell^+/4)!$. Since this upper-bounds the probability
that $\hat{j}' \ge \hat{p}+\ell^+$,
the expected value of $\ell^+$ is $O(1)$.\footnote{Because 
$\sum_{k \ge 1} 1/k! = e-1$.}

An analogous argument holds for $\ell^-$ since $\hat{T}[\hat{i}'\dd \hat{p}-1]$
can also be covered by at most 2 ranges
between consecutive local minima. Adding the expectations
of the contributions $3\ell$ over the $\gamma$ attractor elements, we obtain
$O(\gamma)$.

If the expectation is of the form $c \cdot \gamma$, then at least half
of the permutations produce a grammar of size at most $2c\cdot\gamma$, 
and thus a Las Vegas algorithm finds a permutation producing a grammar
of size at most $2c \cdot \gamma$ after $O(1)$ attempts in expectation. Since
at each attempt we parse $T[1\dd n]$ in time $O(n)$, 
we find a suitable permutation in $O(n)$ expected time provided we know $\gamma$.
\end{proof}

We now perform $O(\log(n/\gamma))$ rounds of locally-consistent
parsing, where the output $T'$ of each round 
is the input to the next. The length of the string halves in each 
iteration, and the grammar grows only by $O(\gamma)$ in each round.

\begin{theorem} \label{thm:rlcfg}
Let $T[1\dd n]$ have an attractor of size $\gamma$. Then there exists a
locally-balanced locally-consistent RLCFG of size 
$g = O(\gamma\log(n/\gamma))$ and height $O(\log(n/\gamma))$ 
that generates (only) $T$, and it can be built in $O(n)$ expected time
and $O(g)$ working space if $\gamma$ is known.
\end{theorem}
\begin{proof}
We apply the grammar construction described in Section~\ref{sec:ec}, which by
Lemmas~\ref{lem:locally-balanced} and \ref{lem:lcg}, is locally balanced and
locally consistent.

We first show that we can build an attractor $\Gamma_r$ for
each $T_r$ formed by $\gamma$ runs of $m_r \in O(1)$ consecutive positions.
This is clearly true for $T_0$, with $m_0=1$. Now assume
this holds for $T_r$. When parsing $T_r$ into blocks to form $T_{r+1}$, each 
run of $m_r$ consecutive attractor positions is parsed into at most $1+\lfloor 
m_r/2 \rfloor$ consecutive symbols $p'$ in $T_{r+1}$, as seen in the proof of
Lemma~\ref{lem:locally-contracting}. Lemma~\ref{lem:4gamma} then 
shows that, if we expand each such mapped attractor positions $p'$ to 
$[p'-1\dd p'+1]$, we obtain an attractor $\Gamma_{r+1}$ for $T_{r+1}$.
The union of the expansions of $1+\lfloor m_r/2 \rfloor$ consecutive positions
$p'$ creates a run of length $m_{r+1} = 3+\lfloor m_r/2 \rfloor$. It then holds
that $\Gamma_{r+1}$ is formed by $\gamma$ runs of at most $m_{r+1}$ positions.

The sequence of values $m_r$ stabilizes. If we solve $m = 3+\lfloor m/2\rfloor$,
we obtain $\lceil m/2 \rceil = 3$. This solves for $m =5$ or $m=6$. Indeed, 
the value is 5 and is reached soon: $m_0=1$, $m_1=3$, $m_2=4$, $m_3 = 5$, 
$m_4 = 5$. Therefore, we safely use $m_r \le 5$ in the following.

The only distinct
blocks in each $T_r$ are those forming $\Gamma_{r+1}$.
Therefore, the parsing of each text $T_r$ produces at most $5\gamma$ distinct
nonterminal symbols.
By \Cref{lem:linear}, we can find in $O(n_r)$ expected time a permutation $\pi_r$
such that the contribution of the $r$th round to the grammar size is $O(|\Gamma_r|)=O(\gamma)$.

The sum of the lengths of all $T_r$s is at most $2n$, thus the total expected
construction cost is $O(n)$. We stop after $r^*=\log(n/\gamma)$ rounds. By then,
$T_{r^*}$ is of length at most $\gamma$ and the cumulative size of the grammar
is $O(\gamma \cdot r^*) = O(\gamma \log(n/\gamma))$. We add a final rule 
$S \rightarrow T_{r^*}$, which adds $\gamma$ to the grammar size. The height
of the grammar is $O(r^*) = O(\log(n/\gamma))$.

As for the working space, at each new round $r$ we generate a permutation $\pi_r$
of $|\Sigma_r|$ cells. Since the alphabet size is a lower bound to the attractor
size, it holds that $|\Sigma_r|\le 5\gamma$. We store the 
distinct blocks that arise during the parsing in a hash table. These are at 
most $5\gamma$ as well, and thus a hash table of size
$O(\gamma)$ is sufficient. The rules themselves, which grow by $O(\gamma)$ in
each round, add up to $O(g)$ total space.
\end{proof}

\subsection{Building the grammar without an attractor}
\label{sec:delta}

Since finding the size $\gamma$ of the smallest attractor is NP-complete
\cite{KP18}, it is interesting that we can find a RLCFG similar to that of
Theorem~\ref{thm:rlcfg} without having to find an attractor nor knowing 
$\gamma$. The key idea is to build on another measure, $\delta$, that
lower-bounds $\gamma$ and is simpler to compute.

\begin{definition} \label{def:delta}
Let $T(\ell)$ be the total number of distinct substrings of length $\ell$ in 
$T$. Then $$\delta = \max \{ T(\ell)/\ell,~\ell \ge 1 \}.$$
\end{definition}

Measure $\delta$ is related to the expression $d_\ell(w)/\ell$, used by \citeN{RRRS13}
to approximate $z$. Analogously to their result \cite[Lem.~4]{RRRS13}, we have the following bound
in terms of attractors.

\begin{lemma} \label{lem:delta}
It always holds $\delta \le \gamma$.
\end{lemma}
\begin{proof}
Since every length-$\ell$ substring of $T$ must have a copy containing
an attractor position, it follows that there are at most $\ell \cdot \gamma$ 
distinct such substrings, that is, $T(\ell)/\ell \le \gamma$ for all $\ell$. 
\end{proof}

\begin{lemma}\label{lem:compute delta}
Measure $\delta$ can be computed in $O(n)$ time and space from $T[1\dd n]$.
\end{lemma}
\begin{proof}
Computing $\delta$ boils down to computing
$T(\ell)$ for all $1 \le \ell \le n$. This is easily computed from a suffix
tree on $T$ \cite{DBLP:conf/focs/Weiner73} (which is built in $O(n)$ time). We first initialize
all the counters $T(\ell)$ at zero. Then we traverse the suffix tree:
for each leaf with string depth $\ell$ we add 1 to $T(\ell)$, and for each 
non-root internal node with $k$ children and string depth $\ell'$ we subtract 
$k-1$ from $T(\ell')$. Finally, for all the $\ell$ values, from $n-1$ to $1$,
we add $T(\ell+1)$ to $T(\ell)$. Thus, the leaves count the unique substrings
they represent, and the latter step accumulates the leaves descending from
each internal node. The value subtracted at internal nodes accounts for the
fact that their $k$ distinct children should count only once toward their
parent.
\end{proof}

We now show that $\delta$ can be used as a replacement of $\gamma$ to build
the grammar.

\begin{theorem} \label{thm:rlcfg2}
Let $T[1\dd n]$ have a minimum attractor of size $\gamma$. Then we can build a
locally-balanced locally-consistent RLCFG of size 
$g = O(\gamma\log(n/\gamma))$ and height $O(\log n)$  that generates (only) $T$ in $O(n)$
expected time and $O(n)$ working space, without knowing $\gamma$.
\end{theorem}
\begin{proof}
We carry out $\log n$ iterations instead of $\log(n/\gamma)$, and the grammar 
is still of size $O(\gamma\log(n/\gamma))$; the extra iterations add only
$O(\gamma)$ to the size. 


The only other place where we need to know $\gamma$ is when applying 
Lemma~\ref{lem:linear}, to check that the total length of the distinct blocks 
resulting from the parsing, using a randomly chosen permutation, is at most
$2c \cdot \gamma$.
A workaround to this problem is to use measure $\delta \le \gamma$, which
(unlike $\gamma$) can be computed efficiently.

To obtain a bound on the sum of the lengths of the blocks formed, we 
add up all the possible substrings multiplied by the probability that they 
become a block.
Consider a substring $\hat{S}[1\dd \ell+3]$ of $\hat{T}$. Whether $\hat{S}$ occurs as a mapped block  
extension, that is, whether it occurs with $\hat{S}' = \hat{S}[3\dd \ell+2]$ being a block, depends only on $\pi$ and 
$\hat{S}$, because by Lemma~\ref{lem:extend}, if $\hat{S}'$ forms a block inside one occurrence of $\hat{S}$,
it must form a block inside each occurrence of $\hat{S}$. Let us now consider
the probability that $\hat{S}'$ forms a block. As in the proof of 
Lemma~\ref{lem:linear}, $\hat{S}[3\dd \ell/2+2]$ must have an increasing sequence of $\pi$-values or $\hat{S}[\ell/2+3\dd \ell+2]$ must have a decreasing sequence of $\pi$-values, and this holds for at most two out of $(\ell/2)!$ permutations $\pi$.

Therefore, any distinct substring of length $\ell+3$ (of which there are 
$T(\ell+3) \le (\ell+3)\delta$) contributes a block of length $\ell$ to the 
grammar size with probability at most $2 / (\ell/2)!$ (note that we may be 
counting the same block several times within different block extensions).
The total expected contribution to the grammar size is therefore
$\sum_{\ell \ge 2} (\ell+3)\delta \cdot \ell \cdot 2 / (\ell/2)! = O(\delta)$.

As in the proof of Lemma~\ref{lem:linear}, given the expectation of the form
$c \cdot \delta$, we can try out permutations until the total contribution to
the grammar size is at most $2c \cdot \delta$. After $O(1)$ attempts, in
expectation, we obtain a grammar of size $O(\delta) \subseteq O(\gamma)$ without
knowing $\gamma$.

We repeat the same process for each text $T_r$, since we know from Theorem~\ref{thm:rlcfg} that every $T_r$ has an attractor of size at most $5\gamma$, so the value $\delta_r$ we compute on $T_r$ satisfies $\delta_r \le 5\gamma$. The sizes of all texts $T_r$ add up to $O(n)$.
\end{proof}



\section{An Index Based on our Grammar} \label{sec:index}

Let $G$ be a locally-balanced RLCFG of $r$ rules and size $g \ge r$ on text 
$T[1\dd n]$, formed with the procedure of Section~\ref{sec:attractor}, thus
$g = O(\gamma\log(n/\gamma))$ with $\gamma$ being the smallest size of an attractor of $T$. We show how to build an index of size $O(g)$ 
that locates the $occ$ occurrences of a pattern $P[1\dd m]$ in time 
$O(m+(occ+1) \log^\epsilon n)$.

We make use of the parse tree and the ``grammar tree'' \cite{CNspire12} of $G$,
where the grammar tree is derived from the parse tree. We extend the concept of
grammar trees to RLCFGs.

\begin{definition} \label{def:grammar-tree} For CFGs,
the {\em grammar tree} is obtained by pruning the parse tree: all but the 
leftmost occurrence of each nonterminal is converted into a leaf and its 
subtree is pruned. Then the grammar tree has exactly one internal node per
distinct nonterminal and the total number of nodes is $g+1$: $r$ internal nodes 
and $g+1-r$ leaves. 
For RLCFGs, we treat rules $A \rightarrow A_1^s$ as $A \rightarrow A_1
A_1^{[s-1]}$, 
where the node labeled $A_1^{[s-1]}$ is always a leaf ($A_1$ 
may also be a leaf, if it is not the leftmost occurrence of $A_1$). Since we
define the size of $A_1^s$ as $2$, the grammar tree is still of size $g+1$.
\end{definition}

We will identify a nonterminal $A$ with the only internal 
grammar tree node labeled $A$. When there is no confusion on the referred node,
we will also identify terminal symbols $a$ with grammar tree leaves.

We extend an existing approach to grammar indexing \cite{CNspire12} to the
case of our RLCFGs. We start by classifying the occurrences in
$T$ of a pattern $P[1\dd m]$ into primary and secondary.

\begin{definition} \label{def:occs}
The leaves of the grammar tree induce a partition of $T$ into $f = g+1-r$ {\em 
phrases}. An occurrence of $P[1\dd m]$ at $T[t\dd t+m-1]$ is {\em 
primary} if the lowest grammar tree node deriving a range of $T$ that contains
$T[t\dd t+m-1]$ is internal (or, equivalently, the occurrence crosses the 
boundary between two phrases); otherwise it is {\em secondary}.
\end{definition}

\subsection{Finding the primary occurrences}
\label{sec:primary}

Let nonterminal $A$ be the lowest (internal) grammar tree node that
covers a primary occurrence $T[t\dd t+m-1]$ of $P[1\dd m]$. Then, if $A
\rightarrow A_1 \cdots A_s$, there exists some $i\in [1\dd s-1]$ and
$q\in [1\dd m-1]$ such that (1) a suffix of $\exp(A_i)$ matches $P[1\dd q]$, and
(2) a prefix of $\exp(A_{i+1}) \cdots \exp(A_s)$ matches $P[q+1\dd m]$.
The idea is to index all the pairs $(\exp(A_i)^{rev},\exp(A_{i+1}) \cdots 
\exp(A_s))$ and find those where the first and second component are prefixed by
$(P[1\dd q])^{rev}$ and $P[q+1\dd m]$, respectively. Note that there is exactly
one such pair per border between two consecutive phrases (or leaves in the
grammar tree). 

\begin{definition} \label{def:locus}
Let $v$ be the lowest 
(internal) grammar tree node that covers a primary occurrence $T[t\dd t+m-1]$ of $P$,
$[t\dd t+m-1] \subseteq proj(v)$. Let $v_i$ be the leftmost child 
of $v$ that overlaps $T[t\dd t+m-1]$, $[t\dd t+m-1] \cap proj(v_i) \not=
\emptyset$. We say that node $v$ is the {\em parent} of the primary occurrence 
$T[t\dd t+m-1]$ of $P$, and node $v_i$ is its {\em locus}.
\end{definition}

We build a multiset $\mathcal{G}$ of $f-1=g-r$ string pairs
containing, for every rule $A\to A_1\cdots A_s$,
the pairs $(\exp(A_i)^{rev}, \exp(A_{i+1})\cdots \exp(A_s))$ for $1\le i < s$.
The $i$th pair is associated with the $i$th child of the (unique) $A$-labeled internal node of the grammar tree.
The multisets $\mathcal{X}$ and $\mathcal{Y}$ are then defined as projections of $\mathcal{G}$ to the first and second coordinate,
respectively. We lexicographically sort these multisets, and represent each pair $(X,Y)\in \mathcal{G}$ by the pair $(x,y)$
of the ranks of $X\in\mathcal{X}$ and $Y\in \mathcal{Y}$, respectively.
As a result, $\mathcal{G}$ can be interpreted as a subset of the two-dimensional integer grid $[1\dd g-r] \times [1\dd g-r]$.

Standard solutions~\cite{CNspire12} to find the primary occurrences consider the partitions $P[1\dd q] \cdot 
P[q+1\dd m]$ for $1\le q<m$. For each such partition, we search for 
$(P[1\dd q])^{rev}$ in $\mathcal{X}$ to find the range $[x_1\dd x_2]$ of symbols
$A_i$ whose suffix matches $P[1\dd q]$, search for $P[q+1\dd m]$ in $\mathcal{Y}$ 
to find the range $[y_1\dd y_2]$ of rule suffixes $A_{i+1}\cdots A_s$ whose
prefix matches $P[q+1\dd m]$, and finally search the two-dimensional grid for 
all the points in the range $[x_1\dd x_2] \times [y_1\dd y_2]$. This retrieves all 
the primary occurrences whose leftmost intersected phrase ends with $P[1\dd q]$. 

From the locus $A_i$ associated with each point $(x,y)$ found, and knowing $q$,
we have sufficient information to report the position in $T$ of this primary 
occurrence and all of its associated secondary occurrences; we describe this
process in Section~\ref{sec:secondary}.

This arrangement follows previous strategies to index CFGs \cite{CNspire12}.
To include rules $A \rightarrow A_1^s$, we just index the pair $(\exp(A_1)^{rev},
\exp(A_1)^{s-1})$, which corresponds precisely to treating the rule as
$A \rightarrow A_1 A_1^{[s-1]}$ to build the grammar tree. It is not necessary 
to index other positions of the rule, since their pairs will look like 
$(\exp(A_1)^{rev},\exp(A_1)^{s'})$ with $s' < s-1$, and if $P[q+1\dd m]$ matches a 
prefix of $\exp(A_1)^{s'}$, it will also match a prefix of $\exp(A_1)^{s-1}$. The 
other occurrences inside $\exp(A_1)^{s-1}$ will be dealt with as secondary 
occurrences.

Finally note
that, by definition, a pattern $P$ of length $m=1$ has no primary
occurrences. We can, however, find all of its occurrences at the end of a
phrase boundary by searching for $P[1\dd 1]^{rev} = P[1]$ in $\mathcal{X}$, to 
find $[x_1\dd x_2]$, and assuming $[y_1\dd y_2] = [1\dd g-r]$. We can only miss the
end of the last phrase boundary, but this is symbol \$, which (just as \#)
is not present in search patterns. We can just treat these points $(x,y)$ as the
primary occurrences of $P$, and report them and their associated secondary
occurrences with the same mechanism we will describe for general patterns
in Section~\ref{sec:secondary}.

A geometric data structure can represent our grid of size $(g-r)\times(g-r)$
with $g-r$ points in $O(g-r) \subseteq O(g)$ space, while performing each range
search in time $O(\log^\epsilon g)$ plus $O(\log^\epsilon g)$ per primary 
occurrence found, for any constant $\epsilon>0$ \cite{DBLP:conf/compgeom/ChanLP11}. 

\subsection{Parsing the pattern} \label{sec:pattern}

In most previous work on grammar-based indexes, all the $m-1$ partitions
$P = P[1\dd q] \cdot \allowbreak P[q+1\dd m]$ are tried out. We now show that, in our locally-consistent parsing, the number of positions that must
be tried is reduced to $O(\log m)$.

\begin{lemma} \label{lem:pattern}
Using our grammar of Section~\ref{sec:attractor}, there are only $O(\log m)$ 
positions $q$ yielding primary occurrences of $P[1\dd m]$.
These positions belong to $M(P)+1$ (see \Cref{def:m}).
\end{lemma}
\begin{proof}
Let $A$ be the parent of a primary occurrence $T[t\dd t+m-1]$, and let $r$ be 
the round where $A$ is formed. There are two possibilities:
\begin{enumerate}
\item $A \rightarrow A_1 \cdots A_s$ is a block-forming rule, and for some
$1 \le i < s$, a suffix of $\exp(A_i)$ matches $P[1\dd q]$, for some $1 \le q <
m$. This means that $q-1 = \min  \hat{B}_{r-1}(t,t+m-1)$. 
\item $A \rightarrow A_1^s$ is a run-length nonterminal, and a suffix of
$\exp(A_1)$ matches $P[1\dd q]$, for some $1 \le q < m$. 
This means that $q-1 = \min B_{r}(t,t+m-1)$.
\end{enumerate}
In either case, $q\in M(P)+1$ by \Cref{lem:m}.
\end{proof}


In order to construct $M(P)$ using \Cref{def:m,def:bp},
we need to already have an occurrence of $P$, which is not feasible in our context.
Hence, we imagine parsing two texts, $T$ and $P^* = \#P\$$, simultaneously
using the permutations $\pi_r$ we choose for $T$ at each round $r$.
It is easy to verify that the results of \Cref{sec:lcp,sec:locality} remain valid
across substrings of both $T$ and $P^*$, because they do not depend on how the 
permutations are chosen. 

Hence, our goal is to parse $P^*$ at query time in order to build $M(P)$ using the occurrence
of $P$ in $P^*$. We now show how to implement this step in $O(m)$ time. 
To carry out the parsing, we must preserve the permutations $\pi_r$ of the 
alphabet used at each of the $O(\log n)$ rounds of the parsing of $T$,
so as to parse $P^*$ in the same way. The alphabets in each round are disjoint
because all the blocks are of length 2 at least. Therefore the total size of 
these permutations coincides with the total number of terminals and 
nonterminals in the grammar, thus by \Cref{lem:linear} and \Cref{thm:rlcfg} they require $O(\gamma)$ space per round and $O(g)$ space overall.

Let us describe the first round of the parsing. We first traverse $P^* = P^*_0$ 
left-to-right and identify the runs $a^\ell$. Those are sought in
a perfect hash table where we have stored all the first-round pairs $(a,\ell)$ 
existing in the text, and are replaced by their corresponding nonterminal
$A \rightarrow a^\ell$ (see below for the case where $a^\ell$ does not appear in the text).
The result of this pass is a new sequence $\hat{P^*} = \hat{P}^*_0$. We then traverse $\hat{P^*}$,
finding the local minima (and thus identifying the blocks) in $O(m)$ time. For this, we have stored the values $\pi(a)=\pi_0(a)$ associated with each terminal $a$ in another perfect hash table (for the nonterminals  $A \rightarrow a^\ell$ just created, we have $\pi(A)=\pi(a)$; recall \Cref{sec:lcp}).
To convert the identified blocks 
$A \rightarrow A_1 \cdots A_k$ into nonterminals for the next round, such tuples $(A_1 \cdots A_k)$ have been stored in yet another perfect hash table, from which the nonterminal $A$ is obtained.
This way, we can identify all the blocks in time $O(m)$, and proceed 
to the next round on the resulting sequence of nonterminals, $P^*_1$.
The size of the first two hash tables is proportional to the number of
terminals and nonterminals in the level, and the size of the tuples stores in
the third table is proportional to the right-hand-sides of the rules created
during the parsing. By \Cref{thm:rlcfg}, those sizes are $O(\gamma)$ per round and $O(g)$ added over all the rounds.

Since the grammar is locally balanced, $P^*$ is parsed in $O(\log m)$ 
iterations, where at the $r$th iteration we parse $P^*_{r-1}$ into a sequence of blocks whose 
total number is at most half of the preceding one, by 
Observation~\ref{obs:length2}. Since we can find the partition into blocks in 
linear time at any given level, the whole parsing takes time $O(m)$. Construction of the sets $B_r(P)$, $\hat{B}_r(P)$,
and $M(P)$ from \Cref{def:bp,def:m}, respectively, also takes $O(m)$ time.

Note that $P^*_r$ might contain blocks and runs that do not occur in $T_r$. By \Cref{lem:lcg}, if a block in $P^*_r$ is not among 
the leftmost 4 or rightmost 2 blocks, then it must also appear within any occurrence of $P$ in $T$, and as a result, the same must also be true for runs in $P^*_{r+1}$.
Consequently, if a block (or a run) is not among those 6 extreme ones yet it does not appear in the hash table, we can abandon the search. 
As for the $O(1)$ allowed new blocks (and runs), we gather them in order to consistently assign new nonterminals and (in case of blocks) arbitrary unused $\pi_r$-values. We then proceed normally with subsequent levels of the parsing. Note that the newly formed blocks cannot appear anymore since distinct levels use distinct symbols, so we do not attempt to insert them into the perfect hash tables.

\subsection{Searching for the pattern prefixes and suffixes}
\label{sec:ztrie}

As a result of the previous section, we need only search for $\tau = O(\log m)$
(reversed) prefixes and suffixes of $P$ in $\mathcal{X}$ or $\mathcal{Y}$,
respectively. In this section we show that the corresponding ranges $[x_1\dd x_2]$
and $[y_1\dd y_2]$ can be found in time $O(m + \tau\log^2 m) = O(m)$. We build on
the following result.

\begin{lemma}[cf.\ \cite{BBPV18,gagie2014lz77,GNP18}] \label{lemma: z-fast}
        Let $\mathcal S$ be a set of strings and assume we have a data
structure supporting extraction of any length-$\ell$ prefix of strings in
$\mathcal S$ in time $f_e(\ell)$ and computation of a given Karp--Rabin
signature $\kappa$ of any length-$\ell$ prefix of strings in $\mathcal S$ in time $f_h(\ell)$.
We can then build a data structure of $O(|\mathcal S|)$ words such that,
later, we can solve the following problem in $O(m + \tau( f_h(m) +\log m ) +
f_e(m))$ time: given a pattern $P[1\dd m]$ and $\tau>0$ suffixes
$Q_1,\dots,Q_\tau$ of $P$, find the ranges of strings in (the
lexicographically-sorted) $\mathcal S$ prefixed by $Q_1,\dots,Q_\tau$.
\end{lemma}
\begin{proof}
The proof simplifies a lemma from \citeN[Lem 5.2]{GNP18}.

        First, we require a Karp--Rabin function $\kappa$ that is collision-free between equal-length text substrings whose length is a power of two. We can find such a function at index construction time in $O(n\log n)$ expected time and $O(n)$ space \cite{DBLP:journals/jda/BilleGSV14}.
        We extend the collision-free property to pairs of equal-letter strings of arbitrary length by switching to the hash function $\kappa'$ defined as $\kappa'(T[i\dd i+\ell-1]) = \langle \kappa(T[i\dd i+2^{\lfloor \log \ell \rfloor}-1]), \kappa(T[i+\ell-2^{\lfloor \log \ell \rfloor}\dd i+\ell-1]), \ell \rangle$.

        Z-fast tries \cite[Sec.~H.2]{BBPV18} solve the \emph{weak} part of the
        lemma in $O(m\log(\sigma)/w + \tau\log m)$ time. They have the same topology of
        a compact trie on $\mathcal{S}$, but use function $\kappa'$ to find a candidate
        node for $Q_i$ in time $O(\log|Q_i|)=O(\log m)$.
        We compute the $\kappa'$-signatures of all pattern suffixes $Q_1, 
        \dots, Q_\tau$ in $O(m)$ time, and then search the z-fast trie
        for the $\tau$ suffixes $Q_i$ in time $O(\tau\log m)$.

        By \emph{weak} we mean that
        the
        returned answer for each suffix $Q_i$ is not
        guaranteed to be correct if $Q_i$ does not prefix any string in $\mathcal S$:
        we could therefore have false positives among the answers, though false
        negatives cannot occur. A procedure for discarding false positives
\cite{gagie2014lz77} requires extracting substrings and their signatures from $\mathcal S$. We describe and simplify this strategy  in detail in order to analyze its time complexity in our scenario.

        Let $Q_1,\dots, Q_j$ be the pattern suffixes for which the z-fast trie found a
        candidate node. Order the pattern suffixes so that $|Q_1| < \dots < |Q_j|$, that is, $Q_i$ is a suffix of $Q_{i'}$ whenever $i<i'$. In addition, let $v_1, \dots, v_j$ be the candidate nodes (explicit or implicit) of the z-fast trie such that all substrings below them are prefixed by $Q_1, \dots, Q_j$ (modulo false positives), respectively, and let $t_i = \mathit{string}(v_i)$ be the substring read from the root of the trie to $v_i$. Our goal is to discard all nodes $v_k$ such that $t_k \neq Q_k$.
        
        Note that it is easy to check (in $O(\tau\cdot f_h(m))$ time) that $\kappa'(Q_i) = \kappa'(t_i)$ for all $i=1, \dots, j$. If a string $t_i$ does not pass this test, then clearly $v_i$ needs to be discarded because it must be the case that $Q_i \neq t_i$. We can thus safely assume that $\kappa'(Q_i) = \kappa'(t_i)$ for all $i=1, \dots, j$.
        
        As a second simplification, we note that it is also easy to check (again in $O(\tau\cdot f_h(m))$ time) that $t_a$ is a suffix of $t_b$ whenever $1 \leq a<b \leq j$. Starting from $a=1$ and $b=2$, we check that $\kappa'(t_a) = \kappa'(t_b[|t_b|-|t_a|+1 \dd |t_b|])$. If the test succeeds, we know for sure that $t_a$ is a suffix of $t_b$, since $\kappa'$ is collision-free among text substrings: we increment $b \leftarrow b+1$, set $a$ to the next index such that $v_a$ was not discarded (at the beginning of the procedure, no $v_a$ has been discarded), and repeat. Otherwise, we clearly need to discard $v_b$ since $\kappa'(Q_b[|t_b|-|t_a|+1 \dd |t_b|]) = \kappa'(Q_a) = \kappa'(t_a) \neq \kappa'(t_b[|t_b|-|t_a|+1 .. |t_b|])$, therefore $Q_b \neq t_b$. Then, we discard $v_b$ and increment $b\leftarrow b+1$. From now on we can thus safely assume that $t_a$ is a suffix of $t_b$ whenever $1 \leq a<b \leq j$.
        
        The last step is to compare explicitly $t_j$ and $Q_j$ in $O(f_e(m))$ time. Since we established that (i) $t_a$ is a suffix of $t_b$ whenever $1 \leq a<b \leq j$, (ii) by definition, $Q_a$ is a suffix of $Q_b$ whenever $1 \leq a<b \leq j$, and (iii) $|Q_i| = |t_i|$ for all $i=1, \dots, j$ (since function $\kappa'$ includes the string's length and we know that $\kappa'(Q_i) = \kappa'(t_i)$ for all $i=1, \dots, j$), checking $t_j = Q_j$ is enough to establish that $t_i = Q_i$ for all $i=1, \dots, j$. 
        However, $t_j \neq Q_j$ is not enough to discard all $v_i$: it could also be the case that only a proper suffix of $t_j$ matches the corresponding suffix of $Q_j$, and some $v_i$ pass the test. We therefore compute the longest common suffix $s$ between $t_j$ and $Q_j$, and discard only those $v_i$ such that $|t_i|>s$.

        To analyze the running time,
        note that we compute $\kappa'$-signatures of strings that are always suffixes of prefixes of length at most $m$ of strings in $\mathcal S$ (because our candidate nodes $v_1, \dots, v_j$ are always at depth at most $m$).
        By definition, to retrieve $\kappa'(t_i)$ we need to compute the two $\kappa$-signatures of the length-$2^e$ prefix and suffix of $t_i$, for some $e\leq \log |t_i| \leq \log m$, $1\leq i \leq j$.
        Computing the required $\kappa'$-signatures reduces therefore to the problem of computing $\kappa$-signatures of suffixes of prefixes of length at most $m$ of strings in $\mathcal S$.
        Let $R' = t_b[|t_b|-s+1\dd |t_b|]$ be such a length-$s$ string of which we need to compute $\kappa(R')$. Then, $\kappa(R') = \kappa(t_b) - \kappa(t_b[1\dd |t_b|-s]) \cdot  c^{s}\ \rm{mod}\ \mu$.
        Both signatures on the right-hand side are prefixes of suffixes of length at most $m$ of strings in $\mathcal S$.
        The value $c^{s}\ \rm{mod}\ \mu$ can moreover be computed in $O(\log m)$ time using the fast exponentiation algorithm.
        It follows that, overall, computing the required $\kappa'$-signatures takes $O(f_h(m) + \log m)$ time per candidate node.
        For the last candidate, we extract the prefix $t_j$ of
        length at most $m$ ($O(f_e(m))$ time) of one of the strings in $\mathcal S$ and compare it with the longest
        candidate pattern suffix ($O(m)$ time). There are at most $\tau$ candidates, so the verification takes time $O(m + \tau\cdot (f_h(m)+\log m) + f_e(m))$. Added to the time to find the candidates in the z-fast trie, we obtain the claimed bounds.
\end{proof}

Therefore,
when $\mathcal{S}$ is $\mathcal{X}$ or $\mathcal{Y}$, we need to extract 
length-$\ell$ prefixes of reverse phrases (i.e., of some $\exp(A_i)^{rev}$) or 
prefixes of consecutive phrases (i.e., of some $\exp(A_{i+1})\cdots \exp(A_s)$)
in time $f_e(\ell)$. The next result implies that we can obtain 
$f_e(\ell)=O(\ell)$.

\begin{lemma}[{cf.\ \cite{DBLP:conf/dcc/GasieniecKPS05},
\cite[Sec.~4.3]{CNspire12}}]\label{lem:extract from rlcfg} Given a RLCFG
of size $g$, there exists a data structure of size $O(g)$ such that any prefix
or suffix of $\exp(A)$ can be obtained from any nonterminal $A$ in
real time.
\end{lemma}
\begin{proof}
\citeN{DBLP:conf/dcc/GasieniecKPS05} show how to extract any prefix of any $\exp(A)$ 
in a CFG of size $g$ in Chomsky Normal Form, in real time, using a data 
structure of size $O(g)$. This was later extended to general CFGs 
\cite[Sec.~4.3]{CNspire12}. We now extend the result to RLCFGs.

Let us first consider prefixes.
Define a forest of tries $T_G$ with one node per distinct nonterminal or
terminal symbol. Let us identify symbols with nodes of $T_G$. Terminal symbols
are trie roots, and $A_1$ is the parent of $A$ in $T_G$ iff $A_1$ is the 
leftmost symbol in the rule that defines
$A$, that is, $A \rightarrow A_1\cdots$. For the rules $A \rightarrow A_1^s$, 
we also let $A_1$ be the parent of $A$. We augment $T_G$ to support
constant-time level ancestor queries \cite{BF04}, which return the ancestor
at a given depth of a given node. To extract $\ell$ symbols
of $\exp(A)$, we start with the node $A$ of $T_G$ and immediately return the
terminal $a$ associated with its trie root (found with a level ancestor query).
We now find the ancestor of $A$ at depth 2 (a child of the trie root). Let
$B$ be this node, with $B \rightarrow a B_2 \cdots B_s$. We recursively
extract the symbols of $\exp(B_2)$ until $\exp(B_s)$, stopping after emitting
$\ell$ symbols. If we obtain the whole $\exp(B)$ and still do not emit $\ell$
symbols, we go to the ancestor of $A$ at depth 3. Let $C$ be this node,
with $C \rightarrow B C_2 \cdots C_r$, then we continue with $\exp(C_2)$,
$\exp(C_3)$, and so on. At the top level of the recursion, we might finally
arrive at extracting symbols from $\exp(A_2)$, $\exp(A_3)$, and so on. 
In this process, when we have to obtain the
next symbols from a nonterminal $D \rightarrow E^s$, we treat it exactly as
$D \rightarrow E \cdots E$ of size $s$, that is, we extract $\exp(E)$ $s-1$
further times.

Overall, we output $\ell$ symbols in time $O(\ell)$. The extraction is not
yet real-time, however, because there may be several returns from the
recursion
between two symbols output. To ensure $O(1)$ time between two consecutive
symbols obtained, we avoid the recursive call for the rightmost child of each
nonterminal, and instead move to it directly.

Suffixes are analogous, and can be obtained in real-time in reverse order
by defining a similar tree $T_G'$ where $A_s$ is the parent of $A$ iff $A_s$ is
the rightmost symbol in the rule that defines $A$, $A \rightarrow \cdots A_s$.
For rules $A \rightarrow A_1^s$, $A_1$ is still the parent of $A$.
\end{proof}

By slightly extending the same structures, we can compute any required 
signature in time $f_h(\ell) = O(\log^2 \ell)$ in our grammars.

\begin{lemma} \label{lem:kr}
In the grammar of Section~\ref{sec:attractor}, we can compute Karp--Rabin 
signatures of prefixes of
length $\ell$ of strings in $\mathcal{X}$ or $\mathcal{Y}$ in time $f_h(\ell) = 
O(\log^2 \ell)$.
\end{lemma}
\begin{proof}
Analogously as for extraction (Lemma~\ref{lem:extract from rlcfg}), we consider
the $O(\log \ell)$ levels of the grammar subtree containing the desired prefix.
For each level, we find in $O(\log \ell)$ time the prefix/suffix of the rule
contained in the desired prefix. Fingerprints of those prefixes/suffixes of
rules are precomputed. 

Strings in $\mathcal{X}$ are reversed expansions of nonterminals. Let every
nonterminal $X$ store the signatures of the reverses of all the suffixes of
$\exp(X)$ that start at $X$'s children. That is, if $X \rightarrow X_1\cdots
X_s$, store the signatures of $(\exp(X_i)\cdots \exp(X_s))^{rev}$ for all $i$.
We use the trie $T'_G$ of the proof of Lemma~\ref{lem:extract from rlcfg}, 
where each trie node is a grammar
nonterminal and its parent is the rightmost symbol of its defining rule. To
extract the signature of the reversed prefix of length $\ell$ of a
nonterminal $X$, we go to the node of $X$ in $T'_G$ and run an exponential
search over its ancestors, so as to find in time $O(\log \ell)$ the lowest one
whose expansion length is $\le \ell$. Let $B$ be that nonterminal, then $B$ is
the first node in the rightmost path of the parse tree from $X$ with $|B| \le
\ell$. Note that the height of $B$ is $O(\log\ell)$ because the grammar is
locally balanced (Lemma~\ref{lem:locally-balanced}), and moreover the parent 
$A \rightarrow B_1 \cdots B_{s-1} B$ of $B$ satisfies $|A| > \ell$. We then 
exponentially search the preceding siblings of $B$
until we find the largest $i$ such that $|B_i|+\cdots+|B| > \ell$ (we must
store these cumulative expansion lengths for each $B_i$). This takes
$O(\log \ell)$ time. We collect the stored signature of
$(\exp(B_{i+1})\cdots \exp(B))^{rev}$; this is part of the signature we will
assemble. Now we repeat the process from $B_i$, collecting the signature from
the remaining part of the desired suffix. Since the depth of the involved nodes
decreases at least by 1 at each step, the whole process takes
$O(\log^2 \ell)$ time. 

The case of $\mathcal{Y}$ is similar, now using the trie $T_G$ of the proof of
Lemma~\ref{lem:extract from rlcfg} and computing prefixes of signatures.
The only difference is that we start from a given child $Y_i$ of a nonterminal
$Y \rightarrow Y_1 \cdots Y_t$ and the signature may span up to the end of
$Y$. So we start with the exponential search for the leftmost $Y_j$ such that
$|Y_i|+\cdots+|Y_j|>\ell$; the rest of the process is similar.

When we have rules of the form $A \rightarrow A_1^s$, we find in constant time
the desired copy $A_i$, from $\ell$ and $|A_1|$. Similarly, we can compute the
signature $\kappa$ of the last $i$ copies of $A_1$ as $\kappa(\exp(A_1)^i) =
\left(\kappa(\exp(A_1))\cdot \frac{c^{|A_1|\cdot i}-1}{c^{|A_1|}-1}\right) \!\!\mod \mu$: 
$c^{|A_1|} \!\!\mod \mu$ and $(c^{|A_1|}-1)^{-1} \!\!\mod \mu$ can be 
stored with $A_1$, and the exponentiation can be computed in 
$O(\log i) \subseteq O(\log\ell)$ time. 
\end{proof}

Overall, we find the $m$ ranges in the grid in time 
$O(m + \tau(f_h(m) +\log m) + f_e(m)) = O(m+\tau \log^2 m) = O(m+\log^3 m) =
O(m)$, as claimed.

\subsection{Reporting secondary occurrences}
\label{sec:secondary}

We report each secondary occurrence in constant amortized time, by adapting
and extending an existing scheme for CFGs \cite{CNspire12} to RLCFGs.
Our data structure enhances the grammar tree with some fields per node $v$
labeled $A$ (where $A$ is a terminal or a nonterminal):
\begin{enumerate}
\item $v.\mathit{anc} = u$ is the nearest ancestor of $v$, labeled $B$, such 
that $u$ is the root or $B$ labels more than one node in the grammar tree. Note
that, since $u$ is internal in the grammar tree, it has the leftmost occurrence
of label $B$ in preorder. This field is undefined in the nodes labeled
$A^{[s-1]}$ we create in the grammar tree (these do not appear in the parse tree).
\item $v.\mathit{offs} = v_i-u_i$, where $proj(v)=[v_i\dd v_j]$ and 
$proj(u)=[u_i\dd u_j]$, is the offset of the projection $\exp(A)$ of $v$ inside 
the projection $\exp(B)$ of $u$. This field is also undefined in the nodes
labeled $A^{[s-1]}$.
\item $v.\mathit{next} = v'$ is the next node in preorder labeled $A$, our 
$null$ if $v$ is the last node labeled $A$ (those next appearances of
$A$ are leaves in the grammar tree). If $B \rightarrow A^s$, the internal
node $u$ labeled $B$ has two children: $v$ labeled $A$ and $v'$ labeled
$A^{[s-1]}$. In this case, $v.next = v'$, and $v'.next$ points to the next
occurrence of a node labeled $A$, in preorder.
\end{enumerate}

Let $u$, labeled $A$, be the parent of 
primary occurrence of $P$, with $A \rightarrow A_1 \cdots A_s$, and $v$, 
labeled $A_i$, be its locus. The grid defined in 
Section~\ref{sec:primary} gives us the pointer to $v$. We then know
that the relative offset of this primary occurrence inside $A_i$ is
$|A_i|-q+1$. We then move to the nearest ancestor of $v$ we have recorded, 
$u' = v.\mathit{anc}$, where the occurrence of $P$ starts at offset 
$\mathit{offs} = |A_i|-q+1+v.\mathit{offs}$ (note that $u'$ can be $u$ or an
ancestor of it). From now on, to find the offset of this occurrence in $T$, we
repeatedly add $u'.\mathit{offs}$ to $\mathit{offs}$ and move to $u' \leftarrow 
u'.\mathit{anc}$. When $u'$ reaches the root, $\mathit{offs}$ is the position 
in $T$ of the primary occurrence. 

At every step of this upward path to the root, we also take the rightward 
path to $u'' \leftarrow
u'.\mathit{next}$. If $u'' \not= null$, we recursively report the copy of the 
primary occurrence inside $u''$, continuing from the same current value of 
$\mathit{offs}$ we have for $u'$.

In other words, from the node $u'=v.\mathit{anc}$ we recursively continue by 
$u'.\mathit{anc}$ and $u'.\mathit{next}$, forming a binary tree of recursive 
calls. All the leaves of this binary tree that are ``left'' children (i.e., by 
$u'.\mathit{anc}$) reach the root of the grammar tree and report a distinct 
offset in $T$ each time. The total number of nodes in this tree is proportional
to the number of occurrences reported, and therefore the amortized cost per 
occurrence reported is $O(1)$.

In case $A \rightarrow A_1^s$, the internal grammar tree node $u$ labeled $A$
has two children: $v$ labeled $A_1$ and $v'=v.\mathit{next}$ labeled 
$A_1^{[s-1]}$. If $P$ has a primary occurrence where $P[1\dd q]$ matches a suffix
of $\exp(A_1)$, the grid will send us to the node $v$, where the occurrence starts
at offset $|A_1|-q+1$. This is just the leftmost occurrence of $P$ within
$\exp(A)$, with offset $|A_1|-q+1$ as well. We must also report all the 
secondary occurrences inside $\exp(A)$, that is, all the offsets 
$i\cdot|A_1|-q+1$, for $i=1,2,\ldots$ as long as 
$i\cdot|A_1|-q+m \le s\cdot |A_1|$. For each such offset we
continue the reporting from $u'=v.\mathit{anc}$, with offset
$\mathit{offs} = i\cdot|A_1|-q+1+v.\mathit{offs}$. 

We might also arrive at such a node $v$ by a $\mathit{next}$ pointer, in which
case the occurrence of $P$ is completely inside $\exp(A_1)$, with offset
$\mathit{offs}$. In this case, we must similarly propagate all the other $s-1$ 
copies of $A_1$ upwards, and then continue to the right. Precisely, we
continue from $u'=v.\mathit{anc}$ and offset $\mathit{offs}+i\cdot|A_1|+
v.\mathit{offs}$, for all $0 \le i < s$. Finally, we continue rightward to
node $v'.\mathit{next}$ and with the original value $\mathit{offs}$.

Our amortized analysis stays valid on these run-length nodes, because we still
do $O(1)$ work per new occurrence reported (these are $s$-ary nodes
in our tree of recursive calls).

\subsection{Short patterns}\label{sec:short patterns}

All our data structures use $O(g)$ space. After parsing the pattern to
find the $\tau = O(\log m)$ relevant cutting points $q$ in time $O(m)$
(Section~\ref{sec:pattern}), and finding the $\tau$ grid ranges $[x_1\dd x_2]
\times [y_1\dd y_2]$ by searching $\mathcal{X}$ and $\mathcal{Y}$ in time $O(m)$
as well (Section~\ref{sec:ztrie}), we look for the primary and secondary
occurrences. Finding the former requires $O(\log^\epsilon g)$ time for
each of the $\tau$ ranges, plus $O(\log^\epsilon g)$ time per primary
occurrence found (Section~\ref{sec:primary}). The secondary occurrences
require just $O(1)$ time each (Section~\ref{sec:secondary}).
This yields total time $O(m+\log m\log^\epsilon g + occ \log^\epsilon g)$ to 
find the $occ$ occurrences of $P[1\dd m]$. 

Next we show how to remove the additive
term $O(\log m \log^\epsilon g)$ by dealing separately with short patterns: we use $O(\gamma)$ further space and leave only an additive 
$O(\log^\epsilon g)$-time term needed for short patterns that do not 
occur in $T$; we then further reduce this term.

The cost $O(\log m \log^\epsilon g)$ comes from the $O(\log m)$ geometric 
searches, each having a component $O(\log^\epsilon g)$ that cannot be charged 
to the primary occurrences found \cite{DBLP:conf/compgeom/ChanLP11}. That cost, however, impacts on
the total search complexity only for short patterns: it can be $\omega(m)$ only
if $m = O(\ell)$, with $\ell=\log^\epsilon g\log\log g$.

We can then store sufficient information to avoid this cost for the 
short patterns.
Since $T$ has an attractor of size $\gamma$, there can be at most $\gamma
\ell$ substrings of length $\ell$ crossing an attractor element, and all the others
must have a copy crossing an attractor element. Thus, there are at most $\gamma \ell$
distinct substrings of length $\ell$ in $T$, and at most $\gamma \ell^2$ 
distinct substrings of length up to $\ell$. We store all these substrings in a 
succinct perfect hash table $H$ \cite{BBD09}, using the function $\kappa'$ of 
Lemma~\ref{lemma: z-fast}
as the key. The associated value for each such substring are the $O(\log \ell)
 = O(\log\log g)$ split points $q$ that are relevant for its search
(Section~\ref{sec:pattern}) {\em and} have points in the corresponding grid 
range (Section~\ref{sec:primary}). Since each partition position $q$ can be
represented in $O(\log\ell) = O(\log\log g)$ bits, we encode all this 
information in $O(\gamma \ell^2 \log^2\ell)$ bits, which is 
$O(\gamma)$ space for any $\epsilon<\frac{1}{2}$. Succinct perfect hash tables 
require only linear-bit space on top of the stored data \cite{BBD09},
$O(\gamma \ell^2)$ bits in our case. Avoiding the partitions that do not
produce any result effectively removes the $O(\log m \log^\epsilon g)$ additive
term on the short patterns, because that cost can be charged to the first 
primary occurrence found. 

Note, however, that function $\kappa'$ is collision-free only among the
substrings of $T$, and therefore there could be short patterns that do not
occur in $T$ but still are sent to a position in $H$ that
corresponds to a short substring of $T$ (within $O(g)$ space we cannot afford 
to store a locus to disambiguate). To discard those patterns, we
proceed as follows. If the first partition returned by $H$ yields no grid points, then this was due to a collision with another pattern, and we
can immediately return that $P$ does not occur in $T$. If, on the other hand,
the first partition does return occurrences, we immediately extract the text 
around the first one in order to verify that the substring is actually $P$.
If it is not, then this is also due to a collision and we return that $P$ does 
not occur in $T$. 

Obtaining the locus $v$ of the first primary occurrence 
from the first partition $q$ takes time $O(\log^\epsilon g)$, and extracting 
$m$ symbols around it takes time $O(m)$, by using 
Lemma~\ref{lem:extract from rlcfg} around $v$.
Detecting that a short pattern $P$ does not
occur in $T$ then costs $O(m+\log^\epsilon g)$. 

We can slightly reduce this cost to  $O(m+\log^\epsilon \gamma)$, as 
follows. Since $g=O(\gamma\log(n/\gamma))$, we have $\log^\epsilon g \in
O(\log^\epsilon \gamma + \log\log(n/\gamma))$. Let $\ell'=\log\log(n/\gamma)$.
We store all the $\gamma\ell'$ distinct text substrings of length 
$\ell'$ in a compact trie $C$,
using perfect hashing to store the children of each node, and associating the
locus $v$ of a primary occurrence with each trie node. The internal trie
nodes represent all the distinct
substrings shorter than $\ell'$. The compact trie $C$ requires 
$O(\gamma \ell') \subseteq O(\gamma\log(n/\gamma))$ space. 
A search for a pattern of length $m \le \ell'$ that does not 
occur in $T$ can then be discarded in $O(m)$ time, by traversing $C$ and
then verifying the pattern around the locus. Thus the additive term 
$O(\log^\epsilon g)$ is reduced to $O(\log^\epsilon\gamma)$.

\subsection{Construction} \label{sec:constr}

\Cref{thm:rlcfg} shows that we can build a suitable grammar in $O(n)$
expected time and $O(g)$ working space, if we know $\gamma$. If not,
\Cref{thm:rlcfg2} shows that the working space rises to $O(n)$.

The grammar tree is then easily built in $O(g)$ time by traversing the grammar 
top-down and left-to-right from the initial symbol, and marking nonterminals 
as we find them for the first time; the next times they are found correspond 
to leaves in the grammar tree, so they are not further explored. By recording
the sizes $|A|$ of all the nonterminals $A$, we also obtain the positions
where phrases start.

Let us now recapitulate the data structures used by our index:
\begin{enumerate}
\item The grid of Section~\ref{sec:primary} where the points of $\mathcal{X}$
and $\mathcal{Y}$ are connected.
\item The perfect hash tables storing the permutations $\pi$, the runs $a^\ell$,
and the blocks generated, for each round of parsing, used in 
Section~\ref{sec:pattern}.
\item The z-fast tries on $\mathcal{X}$ and $\mathcal{Y}$, for
Section~\ref{sec:ztrie}. This includes finding a collision-free Karp--Rabin 
function $\kappa'$.
\item The tries $T_G$ and $T_G'$, provided with level-ancestor queries and
with the Karp--Rabin signatures of all the prefixes and suffixes of
$A_1 \cdots A_s$ for any rule $A \rightarrow A_1 \cdots A_s$. 
\item The extra fields on the grammar tree to find secondary occurrences in
Section~\ref{sec:secondary}.
\item The structures $H$ and $C$ for the short patterns, in
Section~\ref{sec:short patterns}
\end{enumerate}

\citeN[Sec.~4]{NP18} carefully analyze the construction 
cost of points 1 and 3:\footnote{Their $w$ corresponds to our $g$: an upper
bound to the number of phrases in $T$.} 
The multisets $\mathcal{X}$ and $\mathcal{Y}$ can be built from a suffix array
in $O(n)$ time and space, but also from a sparse suffix
array in $O(n\sqrt{\log g})$ expected time and $O(g)$ space \cite{DBLP:conf/soda/GawrychowskiK17}; this
time drops to $O(n)$ if we allow the output to be correct w.h.p.\ only.
A variant of the grid structure of point 1 is built in $O(g\sqrt{\log g})$ 
time and $O(g)$ space \cite{DBLP:conf/soda/BelazzouguiP16}. 
The z-fast tries of point 3 are built in $O(g)$ expected time and space.
However, ensuring that $\kappa'$ is collision-free requires $O(n\log n)$
expected time and $O(n)$ space \cite{DBLP:journals/jda/BilleGSV14}, which is dominant. 
Otherwise, we can build in $O(n)$ expected time and no extra space a signature 
that is collision-free w.h.p.

The structures of point 2 are of total size $O(g)$ and are already built in
$O(g)$ expected time and space during the parsing of $T$. It is an easy 
programming exercise to build the structures of points 4 and 5 in $O(g)$ time; 
the level-ancestor data structure is built in $O(g)$ time as well 
\cite{BF04}.

To build the succinct perfect hash table $H$ of point 6, we traverse the text
around the $g-r$ phrase borders; this is sufficient to spot all the primary
occurrences of all the distinct patterns. There are at most $g\ell^2$ 
substrings of length up to $\ell$ crossing a phrase boundary, where
$\ell = \log^\epsilon g \log\log g$. All their Karp--Rabin signatures $\kappa'$ 
can be computed in time $O(g\ell^2)$ as well, and inserted into a regular hash
table to obtain the $O(\gamma\ell^2)$ distinct substrings. We then build $H$
on the signatures, in $O(\gamma \ell^2)$ expected time \cite{BBD09}. Therefore,
the total expected time to create $H$ is $O(g\ell^2)$, whereas the space is 
$O(\gamma\ell^2)$ (we can obtain this space even without knowing $\gamma$, by 
progressively doubling the size of the hash table as needed).

This construction space can be reduced to $O(\gamma\ell)$ by building a 
separate table $H_m$ for each distinct length $m \in [1\dd \ell]$. Further, since 
we can spend $O(m)$ time when searching for a pattern of length $m$, we can 
split $H_m$ into up to $m$ subtables $H_{m,i}$, which can then be built
separately within $O(g)$ total space: We stop our traversal each time we 
collect $g$ distinct substrings of length $m$, build a separate succinct hash
table $H_{m,i}$ on those, and start afresh to build a new table $H_{m,i+1}$. 
Since there are at most $\gamma m \le g\,\! m$ distinct substrings, we will 
build at most $m$ tables $H_{m,1},\ldots,H_{m,m}$. Note that, in order to 
detect whether each substring appeared previously, we must search all the 
preceding tables $H_{m,1},\ldots,H_{m,i-1}$ for it, which raises the 
construction time to $O(g\ell^3)$. At search time, our pattern may appear in 
any of the $m$ tables $H_{m,i}$, so we search them all in $O(m)$ time.

In order to compute the information on the partitions of each distinct
substring,
we can simulate its pattern search. Since we only need to find its relevant 
split points $q$ (Section~\ref{sec:pattern}), their grid ranges 
(Section~\ref{sec:ztrie}), and which of these are nonempty 
(Section~\ref{sec:primary}), the total time spent per substring of length
up to $\ell$ is $O(\ell + \log \ell \log^\epsilon \gamma) = O(\ell)$.
Added over the up to $\gamma \ell^2$ distinct substrings, the time is 
$O(\gamma\ell^3)$. The whole process then takes
$O(g\ell^3)$ expected time and $O(g)$ space. We enforce 
$\epsilon < \frac{1}{6}$ to keep the time within $O(g \sqrt{\log g})$.

We also build the compact trie $C$ on all the distinct substrings of length 
$\ell' = \log\log(n/\gamma)$. We can collect their signatures $\kappa'$ in 
$O(g\ell')$ time around phrase boundaries, storing them in a temporary hash 
table that collects at most $O(\gamma \ell')$ distinct signatures. For 
each such distinct signature we find, we insert the corresponding substring
in $C$, recording its corresponding locus, in $O(\ell')$ time. The locus
must also be recorded for the internal trie nodes $v$ we traverse, if the 
substring represented by $v$ also crosses the phrase boundary; this must
happen for some descendant leaf of $v$ because $v$ must have a primary
occurrence. Since 
we insert at most $\gamma \ell'$ distinct substrings, the total work on the 
trie is $O(\gamma \ell'\,\!^2)$. Then the expected construction time of $C$
is $O(g\ell' + \gamma\ell'\,\!^2) \subseteq O(g\ell'\,\!^2)
\subseteq O(\gamma\log(n/\gamma)(\log\log(n/\gamma))^2) \subseteq O(n)$. 
The construction space is $O(\gamma\ell') = O(\gamma\log\log(n/\gamma)) 
\subseteq O(\gamma\log(n/\gamma))$.

Note that we need to know $\gamma$ to determine $\ell'$. If we do not know
$\gamma$, we can try out all the lengths, from $\ell' = \log\log(n/g)$ to 
$\log\log n$; note that the unknown correct value is in this range because
$\gamma \le g$. For each length, we build the structures to collect the
distinct substrings of length $\ell$, but stop if we exceed $g$ distinct ones.
Note that we cannot exceed $g$ distinct substrings for
$\ell' \le \log\log(n/\gamma)$ because, in the 
grammar of Section~\ref{sec:attractor}, it holds that $g \ge 
\gamma\log(n/\gamma) \ge \gamma\log\log(n/\gamma) \ge \gamma \ell'$, and this 
is the maximum number of distinct substrings of length $\ell'$ we can produce.
We therefore build the trie $C$ for the value $\ell'$ such that the construction
is stopped for the first time with $\ell'+1$. This value must be $\ell' \ge
\log\log(n/\gamma)$, sufficiently large to ensure the time bounds of 
Section~\ref{sec:short patterns}, and sufficiently small so that the extra 
space is in $O(g)$. The only penalty is that we carry out $\ell'$
iterations in the construction of the hash table (the trie itself is built
only after we find $\ell'$), which costs $O(g\ell'\,\!^2)$ time. This is
the same construction cost we had, but now $\ell'$ can be up to $\log\log n$;
therefore the construction cost is $O(g(\log\log n)^2)$. The construction space
stays in $O(g)$ by design.

\medskip

The total construction cost is then $O(n\log n)$ expected time and $O(n)$
space, essentially dominated by the cost to ensure a collision-free Karp--Rabin 
signature.

\begin{theorem} \label{thm:main}
Let $T[1\dd n]$ have an attractor of size $\gamma$. Then, there exists a data
structure of size $g=O(\gamma\log(n/\gamma))$ that can find the $occ$ 
occurrences of any pattern $P[1\dd m]$ in $T$ in time 
$O(m+\log^\epsilon \gamma + occ\log^\epsilon g) \subseteq 
O(m+(occ+1)\log^\epsilon n)$ for any constant $\epsilon>0$. 
The structure is built in $O(n\log n)$ expected time and $O(n)$ space, 
without the need to know $\gamma$.
\end{theorem}

An index that is correct w.h.p.\ can be built
in $O(n+g\sqrt{\log g}+g(\log\log n)^2) \subseteq O(n+g\sqrt{\log g})$ expected
time. If we know $\gamma$, such an index can be built with $O(\log(n/\gamma))$ expected left-to-right passes on $T$ (to build the grammar) plus $O(\gamma\log(n/\gamma))$ main-memory space.

\medskip
Finally, note that if we want to report only $k < occ$ occurrences of $P$, their locating time does not anymore amortize to $O(1)$ as in Section~\ref{sec:secondary}. Rather, extracting each occurrence requires us to climb up the grammar tree up to the root. In this case, the search time becomes $O(m+(k+1)\log n)$.

\subsection{Optimal search time} \label{sec:optimal}

We now explore various space/time tradeoffs for our index, culminating with a
variant that achieves, for the first, time, optimal search time within space
bounded by an important family of repetitiveness measures. The tradeoffs are
obtained by considering other data structures for the grid of 
Section~\ref{sec:primary} and for the perfect hash tables of 
Section~\ref{sec:short patterns}.
Table~\ref{tab:tradeoffs} summarizes the results in a slightly simplified form;
the construction times stay as in \Cref{thm:main}.

\begin{table}[t]
\begin{center}
\tbl{Space-time tradeoffs within attractor-bounded space; formulas
are slightly simplified.
\label{tab:tradeoffs}}
{\begin{tabular}{l|c|c}
Source & Space & Time \\
\hline
Baseline \cite{NP18} 
	& $O(\gamma\log(n/\gamma))$ 
	& $O(m\log n + occ \log^\epsilon n)$ \\
\hline
\Cref{thm:main}            
	& $O(\gamma\log(n/\gamma))$ 
	& $O(m + (occ+1) \log^\epsilon n)$ \\
Corollary~\ref{cor:tradeoff1}            
	& $O(\gamma\log n)$ 
	& $O(m + occ \log^\epsilon n)$ \\
Corollary~\ref{cor:tradeoff2}            
	& $O(\gamma\log(n/\gamma)\log\log n)$ 
	& $O(m + (occ+1) \log\log n)$ \\
Corollary~\ref{cor:tradeoff3}            
	& $O(\gamma\log n\log\log n)$ 
	& $O(m + occ \log\log n)$ \\
\Cref{thm:opt}            
	& $O(\gamma\log(n/\gamma)\log^\epsilon n)$ 
	& $O(m+occ)$ \\
\hline
\end{tabular}}
\end{center}
\end{table}

A first tradeoff is obtained by discarding the table $H$ of
Section~\ref{sec:short patterns} and using only a compact trie $C'$, now to 
store the locus of a primary occurrence and the relevant split points of each 
substring of length up to 
$\ell = \log^\epsilon g \log\log g$. This adds $O(\gamma\ell)$ to the space, 
but it allows verifying that the short patterns actually occurs in $T$ in time 
$O(m)$ without using the grid. As a result, the additive term 
$O(\log^\epsilon \gamma)$ disappears from the search time.

As seen in Section~\ref{sec:constr}, the extra construction time for $C'$ is
now $O(g\ell^2)$, plus $O(\gamma\ell^3)$ to compute the relevant split points.
This is within the $O(g\ell^3)$ time bound obtained for
\Cref{thm:main}. The construction space is $O(\gamma\ell)$,
which we can assume to be $O(n)$ because it is included in the final
index size; if this is larger than $n$ then the result holds trivially by
using instead a suffix tree on $T$.

\begin{corollary} \label{cor:tradeoff1}
Let $T[1\dd n]$ have an attractor of size $\gamma$. Then, there exists a data
structure of size
$g=O(\gamma(\log(n/\gamma)+\log^\epsilon(\gamma\log(n/\gamma))\log\log(\gamma\log(n/\gamma)))) \subseteq O(\gamma \log n)$ that can find the $occ$ 
occurrences of any pattern $P[1\dd m]$ in $T$ in time 
$O(m+occ\log^\epsilon g) \subseteq O(m+occ\log^\epsilon n)$ for any constant
$\epsilon>0$. 
The structure is built in $O(n\log n)$ expected time and $O(n)$ space, 
without the need to know $\gamma$.
\end{corollary}

By using $O(g\log\log g)$ space for the grid, the range queries run in time
$O(\log\log g)$ per query and per returned item \cite{DBLP:conf/compgeom/ChanLP11}. This reduces the
query time to $O(m+\log m \log\log g + occ \log\log g)$, which can be further
reduced with the same techniques of Section~\ref{sec:short patterns}: The
additive term can be relevant only if $m = O(\ell)$ with $\ell =
\log\log g \log\log\log g$. We then store in $H$ all the
$\gamma \ell^2$ patterns of length up to $\ell$, with their relevant 
partitions, using $O(\gamma \ell^2 (\log\ell)^2) =
O(\gamma (\log\log g)^2 (\log\log\log g)^4)$ bits, which is $O(\gamma)$
space. We may still need $O(\log\log g)$ time to determine that a short pattern
does not occur in $T$. By storing the patterns of length $\ell' =
\log\log\log(n/\gamma)$ in trie $C$, this time becomes $O(\log\log\gamma)$.

The grid structure can be built in time $O(g\log g)$. 
The construction time for $H$ and $C$ is lower than in
Section~\ref{sec:constr}, because $\ell$ and $\ell'$ are smaller here.

\begin{corollary} \label{cor:tradeoff2}
Let $T[1\dd n]$ have an attractor of size $\gamma$. Then, there exists a data
structure of size $g=O(\gamma\log(n/\gamma)\log\log(\gamma\log(n/\gamma)))
\subseteq O(\gamma \log (n/\gamma)\log\log n)$ that can find the $occ$ 
occurrences of any pattern $P[1\dd m]$ in $T$ in time 
$O(m+\log\log\gamma+occ\log\log g) \subseteq O(m+(occ+1)\log\log n)$.
The structure is built in $O(n\log n)$ expected time and $O(n)$ space, 
without the need to know $\gamma$.
\end{corollary}

By discarding $H$ and building $C'$ on the substrings of length
$\ell=\log\log g \log\log\log g$, we increase the space by 
$O(\gamma\ell^2)$ and remove the additive term in the search time.
The construction time for the grid is still $O(g\log g)$, but that of $C$ is 
within the bounds of Corollary~\ref{cor:tradeoff1}, because $\ell$ is smaller 
here.

\begin{corollary} \label{cor:tradeoff3}
Let $T[1\dd n]$ have an attractor of size $\gamma$. Then, there exists a data
structure of size $g=O(\gamma(\log(n/\gamma)\log\log(\gamma\log(n/\gamma))+
(\log\log(\gamma\log(n/\gamma))\log\log\log(\gamma\log(n/\gamma)))^2))
\subseteq O(\gamma \log n\log\log n)$ that can find the $occ$ 
occurrences of any pattern $P[1\dd m]$ in $T$ in time 
$O(m+occ\log\log g) \subseteq O(m+occ\log\log n)$.
The structure is built in $O(n\log n)$ expected time and $O(n)$ space, 
without the need to know $\gamma$.
\end{corollary}

Finally,
a larger geometric structure \cite{DBLP:conf/focs/AlstrupBR00} uses $O(g\log^\epsilon g)$ space, for 
any constant $\epsilon>0$, and reports in $O(\log\log g)$ time per query and 
$O(1)$ per result. This yields $O(m+\log m\log\log g+occ)$ search time.
To remove the second term, we again index all the patterns of length $m\le\ell$,
for $\ell = \log\log g \log\log\log g$, of which there are at most 
$\gamma \ell^2$. Just storing the relevant split points $q$ is not sufficient 
this time, however, because we cannot even afford the $O(\log\log g)$ time to 
query the nonempty areas. 

Still, note that the search time can be written as $O(m+\ell+occ)$. Thus, we
only care about the short patterns that, in addition, occur less than $\ell$ 
times, since otherwise the third term, $O(occ)$, absorbs the second. Storing 
all the occurrences of such patterns requires $O(\gamma\ell^2)$ space: An
enriched version $C''$ of the compact trie $C$ records the number of
occurrences in $T$ of each node. Only the leaves 
(i.e., the patterns of length exactly $\ell$) store their occurrences (if they are
at most $\ell$). Since there are at most $\gamma \ell$ leaves, the total
space to store those occurrences is $O(\gamma \ell^2)$, dominated by the grid
size. Shorter patterns 
correspond to internal trie nodes, and for them we must traverse all the 
descendant leaves in order to collect their occurrences. 

To handle a pattern $P$ of length up to $\ell$, then, we traverse $C''$ and
verify $P$ around its locus. If $P$ occurs in $T$, we see if the trie node
indicates it occurs more than $\ell$ times. If it does, we use the normal
search procedure using the geometric data structure and propagating the
secondary occurrences. Otherwise, its (up to $\ell$) occurrences are obtained
by traversing all the leaves descending from its trie node: if an internal node
occurs less than $\ell$ times, its descendant leaves also occur less than
$\ell$ times, so all the occurrences of the internal node are found in the 
descendant leaves. The search time is then always $O(m+occ)$.

The expected construction time of the geometric structure \cite{DBLP:conf/focs/AlstrupBR00}
is $O(g \log g)$, and its construction space is $O(g\log^\epsilon g)$. 
Note that if the construction space exceeds $O(n)$, then so does the size of 
our index. In this case, a suffix tree obtains linear construction time and
space with the same search time. Thus, we can assume the construction
space is $O(n)$.

The trie $C''$ is not built in the same way $C$ is built 
in Section~\ref{sec:constr}, because we need to record the number of
occurrences of each string of length up to $\ell$. We slide the window of 
length $\ell$ through the whole text $T$ instead of only around phrase 
boundaries. We maintain the distinct signatures $\kappa'$ found in a regular 
hash table, with the counter of how many times they appear in $T$. When a new 
signature appears, its string is inserted in $C''$, a pointer from the hash 
table to the 
corresponding trie leaf is set, and the list of occurrences of the substring 
is initialized in the trie leaf, with its first position just found. Further 
occurrence positions of the string are collected at its trie leaf, until they 
exceed $\ell$, in which case they are deleted. Thus we spend $O(n)$ expected 
time in the hash table and collecting occurrences, plus $O(\gamma \ell^2)$ 
time inserting strings in $C''$. From the number of occurrences of each 
leaf we can finally propagate those counters upwards in the trie, in 
$O(\gamma\ell)$ additional time.

\begin{theorem} \label{thm:opt}
Let $T[1\dd n]$ have an attractor of size $\gamma$. Then, there exists a data
structure of size $O(\gamma\log(n/\gamma)\log^\epsilon(\gamma\log(n/\gamma)))
\subseteq O(\gamma\log(n/\gamma)\log^\epsilon n)$, for any constant 
$\epsilon>0$, that can find the $occ$ 
occurrences of any pattern $P[1\dd m]$ in $T$ in time $O(m+occ)$.
The structure is built in $O(n\log n)$ expected time and 
$O(n)$ space, without the need to know $\gamma$.
\end{theorem}


\section{Counting Pattern Occurrences} \label{sec:counting}

\citeN{Nav18} shows how an index like the one we describe in
Section~\ref{sec:index} can be used for counting the number of occurrences of
$P[1\dd m]$ in $T$. First, he uses the result of \citeN{DBLP:journals/siamcomp/Chazelle88} that a $p
\times p$ grid can be enhanced by associating elements of any algebraic 
semigroup to the points, so that later we can aggregate all the elements inside 
a rectangular area in time $O(\log^{2+\epsilon} p)$, for any constant 
$\epsilon>0$, with a structure using $O(p)$ space.\footnote{\citeN{Nav18} gives a simpler
explicit construction for groups.}
The structure is built in $O(p\log p)$ time and $O(p)$ space \cite{DBLP:journals/siamcomp/Chazelle88}.
Then, \citeN{Nav18} shows that one can associate with a CFG the number of secondary 
occurrences triggered by each point in a grid analogous to that of
Section~\ref{sec:primary}, so that their sums can be computed as described.

We now improve upon the space and time using our RLCFG of
Section~\ref{sec:index}. 
Three observations are in order (cf.~\cite{CNspire12,Nav18}):
\begin{enumerate}
\item The occurrences reported are all those derived from each point $(x,y)$
contained in the range $[x_1\dd x_2] \times [y_1\dd y_2]$ of each relevant 
partition $P[1\dd q] \cdot P[q+1\dd m]$.
\item Even if the same point $(x,y)$ appears in distinct overlapping ranges 
$[x_1\dd x_2] \times [y_1\dd y_2]$, each time it corresponds to a distinct value of 
$q$, and thus to distinct final offsets in $T$. Therefore, all the occurrences 
reported are distinct.
\item The number of occurrences reported by our procedure in 
Section~\ref{sec:secondary} depends only on the initial locus
associated with the grid point $(x,y)$. This will change with run-length nodes
and require special handling, as seen later.
\end{enumerate}

Therefore, we can associate with each point $(x,y)$ in the grid (and with the corresponding primary occurrence) the total
number of occurrences triggered with the procedure of Section~\ref{sec:secondary}. 
Then, counting the number of occurrences of a partition $P = P[1\dd q] \cdot P[q+1\dd m]$ corresponds 
to summing up the number of occurrences of the points that lie in the appropriate range of the grid.

As seen in Section~\ref{sec:pattern}, with our particular grammar there are only
$O(\log m)$ partitions of $P$ that must be tried in order to recover all of its
occurrences. Therefore, we use our structures of Sections~\ref{sec:primary} to
\ref{sec:ztrie} to find the $O(\log m)$ relevant ranges
$[x_1\dd x_2] \times [y_1\dd y_2]$, all in $O(m)$ time, and then we count the number
of occurrences in each such range in time $O(\log^{2+\epsilon} p)\subseteq 
O(\log^{2+\epsilon}g)$. The total counting time is then 
$O(m+\log m \log^{2+\epsilon} g)$. When the second term dominates, $m \le 
\log m \log^{2+\epsilon} g$, it holds $\log m \log^{2+\epsilon} g \in 
O(\log^{2+\epsilon} g \log\log g)$, which is $O(\log^{2+\epsilon} g)$ by 
infinitesimally adjusting $\epsilon$.

Under the assumption that there are no run-length rules (we remove this assumption later), our counting time is then $O(m+\log^{2+\epsilon} g)$. This improves sharply
upon the previous result \cite{Nav18} in space (because it builds
the grammar on a Lempel--Ziv parse instead of on attractors) and in time
(because it must consider all the $m-1$ partitions of $P$).

To build the structure, we must count the number of secondary occurrences 
triggered from any locus $v$, and then associate it with every point $(x,y)$
having $v$ as its locus. More precisely, we will compute the number of times
any node $u$ occurs in the parse tree of $T$. The process corresponds to 
accumulating occurrences
over the DAG defined by the pointers $u.\mathit{anc}$ and $u.\mathit{next}$
of the grammar tree nodes $u$. Initially, let the counter be $c(u)=0$ for every
grammar tree node $u$, except the root, where $c(root)=1$. We now traverse all
the nodes $u$ in some order, calling $\mathit{compute}(u)$ on each.
Procedure $\mathit{compute}(u)$ proceeds as follows: If $c(u)>0$ then the 
counter is already computed, so it simply returns $c(u)$. Otherwise, it sets
$c(u) = \mathit{compute}(u.\mathit{anc})+\mathit{compute}(u.\mathit{next})$,
recursively computing the counters of the two nodes. Nodes $A\rightarrow A_1^s$ are special cases. If $u.\mathit{next}$ is of the form 
$A_1^{[s-1]}$, then the correct formula is 
$c(u) = s \cdot \mathit{compute}(u.\mathit{anc})+
        \mathit{compute}(u.\mathit{next}.\mathit{next})$.
On the other hand, we do nothing for $\mathit{compute}(u)$ if $u$ is of the
form $A_1^{[s-1]}$.
The total cost is
the number of edges in the DAG, which is 2 per grammar tree node, $O(g)$.

Finally, the counter of each point $(x,y)$ associated with locus node $v$ is
the value $c(u)$, where $u$ is the parent of $v$. 
A special case arises, however, if $u$ corresponds to a run-length node 
$A \rightarrow A_1^s$, in which case
the locus $v$ is $A_1$. As seen in Section~\ref{sec:secondary}, the
number of times $u$ is reported is $s-\lceil (m-q)/|A_1| \rceil$, and
therefore the correct counter to associate with $(x,y)$ is 
$(s-\lceil (m-q)/|A_1| \rceil) \cdot c(u)$.
The problem is that such a formula depends on $m-q$, so each point $(x,y)$
could contribute differently for each alignment of the pattern. 
We then take a different approach for counting these occurrences.

Associated with loci $A_1$ with parent $A \rightarrow A_1^s$, instead of $(x,y)$, we add to the grid the points $(x,y')=(\exp(A_1)^{rev},\exp(A_1))$ with weight $c(u)$
and $(x,y'')=(\exp(A_1)^{rev},\exp(A_1)^2)$ with weight $(s-2)c(u)$,
extending the set $\mathcal{Y}$ so that it contains both $\exp(A_1)$ and $\exp(A_1)^2$.
(Note that there could be various equal 
string pairs, which can be stored multiple times, or we can accumulate their
counters.)
We distinguish three cases.
\begin{enumerate}[(i)]
\item For the occurrences
where $P[q+1\dd m]$ lies inside $\exp(A_1)$ (i.e., $m-q \le |A_1|$), 
the rule $A\to A_1^s$ is counted $c(u)+(s-2)c(u) = (s-1)c(u)$ times because both $(x,y')$
and $(x,y'')$ are in the range queried.

\item For the
occurrences where $P[q+1\dd m]$ exceeds the first $\exp(A_1)$ but does not span 
more than two (i.e., $|A_1| < m-q \le 2|A_1|$), 
the rule $A\to A_1^s$ is counted $(s-2)c(u)$ times because $(x,y'')$ is in the range queried
but $(x,y')$ is not.

\item For the occurrences where $P[q+1\dd m]$ spans more than two copies of $\exp(A_1)$,
however, the rule $A\to A_1^s$ is not counted at all because neither $(x,y')$ nor $(x,y'')$ is in the range queried.
\end{enumerate}

The key to handle the third case is that, if $P[1\dd q]$ spans a suffix of $\exp(A_1)$ and $P[q+1\dd m]$ spans at least two 
consecutive copies of $\exp(A_1)$, then it is easy to see that $P$ is 
``periodic'', $|A_1|$ being a ``period'' of $P$ \cite{Jewels}.

\begin{definition}
A string $P[1\dd m]$ has a {\em period} $p$ if $P$ consists of $\lfloor m/p 
\rfloor$ consecutive copies of $P[1\dd p]$ plus a (possibly empty) prefix of 
$P[1\dd p]$. Alternatively, $P[1\dd m-p]=P[p+1\dd m]$. The string $P$ is {\em 
periodic} if it has a period $p \le m/2$.
\end{definition}

We next show an important property relating periods and run-length nodes.

\begin{lemma} \label{lem:period}
Let there be a run-length rule $A \rightarrow A_1^s$ in our grammar. Then
$|A_1|$ is the shortest period of $\exp(A)$.
\end{lemma}
\begin{proof}
Consider an $A$-labeled node $v$ in the parse tree of $T$
and let $proj(v)=[i\dd j]$ so that $T[i\dd j]=\exp(A)$.
Denote the shortest period of $\exp(A)$ by $p$ and note that $|A_1|$ is also a period of $\exp(A)=\exp(A_1)^s$.
We conclude from the Periodicity Lemma~\cite{fine1965uniqueness} that $p = \gcd(p, |A_1|)$ and thus $d = |A_1|/p$ is an integer. For a proof by contradiction, suppose that $d>1$. Let $r$ denote the level of the run represented by $v$ (so that $A$ is a symbol in $\hat{T}_r$ and $A_1$ is a symbol in $T_r$).
\begin{claim}
For each level $r'\in [0\dd r]$, both $i+p-1$ and $j-p$ are level-$r'$ block boundaries.
\end{claim}
\begin{proof}
We proceed by induction on $r'$. The base case for $r'=0$ holds trivially.
Thus, consider a level $r'\in [1\dd r]$ and suppose that the claim holds for $r'-1$.
By the inductive assumption, $T[i+p\dd j]=T[i\dd j-p]$ consist of full level-$(r'-1)$ blocks, so \Cref{lem:altr} yields $\hat{B}_{r'-1}(i+p,j)=\hat{B}_{r'-1}(i,j-p)$.
Since $i+dp-1$ is a level-$r'$ block boundary, this set is non-empty and its minimum satisfies $\min \hat{B}_{r'-1}(i+p,j) < dp-p$.
The final claim of \Cref{lem:altr} thus yields $B_{r'}(i+dp-p,j-p) = B_{r'}(i+dp,j)$. 
Consequently, since $i+dp-1$ is a level-$r'$ block boundary, $p-1\in B_{r'}(i+dp-p,j-p)=B_{r'}(i+dp,j)$,
so $i+dp+p-1$ is also a level-$r'$ block boundary. Iterating this reasoning $d(s-1)-2$ more times, we conclude that
$i+dp+2p-1, i+dp+3p-1,\ldots, j-p$ are all level-$r'$ block boundaries. 
Moreover, \Cref{lem:altr} applied to $T[i\dd j-dp]=T[i+dp\dd j]$,
which consist of full level-$r'$ blocks, implies $p-1 \in B_{r'}(i,j-dp)=B_{r'}(i+dp,j)$, so $i+p-1$ is also a level-$r'$ block boundary.
\end{proof}
Note that $T[i\dd j]$ consists of $s$ full level-$r$ blocks of length $dp$ each.
The claim instantiated to $r'=r$ contradicts this statement imposing blocks of
length at most $p$ at the extremities. \qed
\end{proof}

\Cref{lem:period} implies that, in the remaining case to be handled, the length $|A_1|$ must be precisely the shortest period of $P$.

\begin{lemma}
Let $P$ be contained in $\exp(A)$ and contain two consecutive copies of $\exp(A_1)$, from rule $A \rightarrow
A_1^s$. Then $|A_1|$ is the shortest period of $P$.
\end{lemma}
\begin{proof}
Clearly $|A_1|$ is a period of $P$ because $P[1\dd m]$ is contained in a 
concatenation of strings $\exp(A_1)$; further, $|A_1|\le m/2$. Now assume $P$
has a shorter period, $p < |A_1|$. Since $|A_1|+p < m$, $P$ also has a period
of length $p' = \gcd(|A_1|,p)$ \cite{fine1965uniqueness}. This period is smaller 
than $|A_1|$ and divides it. Since $P$ contains $\exp(A_1)$, this implies that 
$\exp(A_1)$, and thus $\exp(A)$, also have a period $p' < |A_1|$, contradicting
Lemma~\ref{lem:period}.
\end{proof}

Therefore, all the run-length nonterminals $A \rightarrow A_1^s$, where $A_1$ 
is a locus of $P$ with offset $q$ and
$m \ge 2|A_1|$, must satisfy $\exp(A_1)=P[q+1\dd q+p]$, where $p$ is the
shortest period of $P$. The shortest period $p$ is easily computed in $O(m)$ 
time \cite[Sections~1.7 and 3.1]{Jewels}.

It is therefore sufficient to compute the Karp--Rabin fingerprints
$k=\kappa'(\exp(A_1))$ (which we 
easily retrieve from the data we store for Lemma~\ref{lem:kr})
for all the run-length rules $A \rightarrow A_1^s$, and store them in a
perfect hash table with information on $A_1$. 
Let $s(A_1) = \{ s \ge 3, A \rightarrow A_1^s\}$ be the different exponents associated with $A_1$. To each $s \in s(A_1)$, we associate two values 
\[
        c(A_1,s) = \sum \{ c(A) : A \rightarrow A_1^{s'}, s'\geq s \} \quad \text{and} \quad
        c'(A_1,s) = \sum \{ s'\cdot c(A) : A \rightarrow A_1^{s'}, s'\geq s\}.
\]
where $c(A)$ refers to $c(u)$ for the (only) internal grammar tree node $u$ corresponding to nonterminal $A$.
The total space to store the sets $s(A_1)$ and associated values is $O(g)$. 

For each of the $O(\log m)$ relevant splits 
$P[1\dd q] \cdot P[q+1\dd m]$ obtained in Section~\ref{sec:pattern}, if $m-q > 2p$, then
we look for $k=\kappa'(P[q+1\dd q+p])$ in the hash table. If we find it mapped to a non-terminal $A_1$, 
then we add $c'(A_1,s_{\min})-c(A_1,s_{\min})\lceil (m-q)/p \rceil$ to the result, 
where $s_{\min} = \min \{ s \in s(A_1), (s-1)|A_1|\ge m-q\}$.
This ensures that each rule $A\to A_1^{s}$ with $s\ge 3$
and $|A_1|(s-1)\ge m-q$ is counted $(s-\lceil (m-q)/p \rceil)\cdot c(A)$ times. We find
$s_{\min}$ by exponential search
on $s(A_1)$ in $O(\log m)$ time, which over all the splits adds up to $O(\log^2 m)$.

Note that all the Karp--Rabin fingerprints
for all the substrings of $P$ can be computed in $O(m)$ time (see
Section~\ref{sec:ztrie}), and that we can easily rule out false positives:
\Cref{lemma: z-fast} filters out any decomposition of $P$ for which $P[q+1\dd m]$ is not a prefix
of any string $y\in \mathcal{Y}$. Since $\exp(A_1)^{s-1}\in \mathcal{Y}$ for every rule $A\to A_1^s$
and since $\mathcal{Y}$ consists of substrings of $T$, this guarantees that $\kappa'$
does not admit any collision between $P[q+1\dd q+p]$ and a substring of $T$.

\begin{theorem} \label{thm:count}
Let $T[1\dd n]$ have an attractor of size $\gamma$. Then, there exists a data
structure of size $g = O(\gamma\log(n/\gamma))$ that can count the number of
occurrences of any pattern $P[1\dd m]$ in $T$ in time 
$O(m+\log^{2+\epsilon} g) \subseteq O(m+\log^{2+\epsilon} n)$ for 
any constant $\epsilon>0$.
The structure can be built in $O(n\log n)$ expected time and $O(n)$ space,
without the need to know $\gamma$.
\end{theorem}

An index that is correct w.h.p.\ can be built
in $O(n+g\log g)$ expected time (the structures for secondary occurrences
and for short patterns, Sections~\ref{sec:secondary} and 
\ref{sec:short patterns}, are not needed).
If we know $\gamma$, the index can be built in $O(\log(n/\gamma))$ expected left-to-right passes on $T$ plus $O(g)$ main memory space.

\subsection{Optimal time}

\citeN{DBLP:journals/siamcomp/Chazelle88} offers other tradeoffs for operating the elements in a range, all very similar and with the 
same construction cost: $O(\log^2 p \log\log p)$ time and $O(p\log\log p)$ 
space, $O(\log^2 p)$ time and $O(p\log^\epsilon p)$ space. These yield, for our
index, $O(m+(\log n\log\log n)^2)$ time and $O(g\log\log g)$ space, and
$O(m+\log^2 n\log\log n)$ time and $O(g\log^\epsilon g)$ space.

If we use $O(p\log p)$ space, however, the cost to compute the 
sum over a range decreases significantly, to $O(\log p)$ \cite{Wil85,DBLP:conf/focs/AlstrupBR00}. The
expected construction cost becomes $O(p\log^2 p)$ \cite{DBLP:conf/focs/AlstrupBR00}.
Therefore, using $O(g\log g) \subseteq O(\gamma\log(n/\gamma)\log n)$ space, 
we can count in time $O(m+\log m\log g) \subseteq O(m+\log g\log\log g)
\subseteq O(m+\log n\log\log n)$, which is yet another tradeoff.

More interesting is that we can reduce this latter time to the optimal $O(m)$.
We index in a compact trie like $C''$ of Section~\ref{sec:optimal} all the text
substrings of length up to $\ell=2\log n \log(n/\gamma)$, directly storing their
number of occurrences (but not their occurrence lists as in $C''$). Since there
are $\gamma \ell$ distinct substrings of length $\ell$, this requires 
$O(\gamma \log n \log(n/\gamma))$ space. 


Consider our counting time $O(m+\log m\log n)$.
If $\log(n/\gamma) \le \log\log n$, then $\gamma \ge n/\log n$, and thus a 
suffix tree using space $O(n) = O(\gamma\log n)$ can count in optimal time 
$O(m)$. Thus, assume $\log(n/\gamma) > \log\log n$. The counting time can 
exceed $O(m)$ only if $m \le \log m \log n$.  In this case,
since $m \le \log m \log n \le \log^2 n$, we have $m \le 2\log n \log\log n \le
2\log n \log(n/\gamma) = \ell$. All the queries for patterns of those lengths 
are directly answered using our variant of $C''$, in time $O(m)$, and thus our
counting time is always $O(m)$.

We can still apply this idea if we do not know $\gamma$. Instead, we compute
$\delta$ (recall Section~\ref{sec:delta}) and use $\ell = 2\log
n\log(n/\delta)$. Since there are $T(\ell) \le \delta \ell$ distinct
substrings of length $\ell$ in $T$, the space for $C''$ is $O(\delta \ell)
= O(\delta \log n \log(n/\delta)) \subseteq O(\gamma \log n \log(n/\gamma))$,
the latter by Lemma~\ref{lem:delta}. 
The reasoning of the previous paragraph
then applies verbatim if we replace $\gamma$ by $\delta$.

The total space is then $O(g\log g + \gamma \log n \log(n/\gamma)) =
O(\gamma \log n \log(n/\gamma))$. The construction cost of $C''$ is 
$O(n+\gamma\log^2 n \log^2(n/\gamma))$ time and $O(\gamma\log n\log(n/\gamma))$
space.\footnote{If we use
$\ell=2\log n\log(n/\delta)$, then $C''$ is built in $O(\delta\log^2 n
\log^2(n/\delta)) \subseteq O(\gamma\log^2 n \log^2(n/\gamma))$ time and
$O(\delta\log n \log(n/\delta)) \subseteq O(\gamma\log n \log(n/\gamma))$ 
space, because the costs increase with $\delta$.}
Alternatively we can obtain it by pruning the suffix tree of $T$ in
time and space $O(n)$. 
The cost to build the grid is $O(g\log^2 g) \subseteq (g\log^2 n)$.
Note that, if $\gamma\log(n/\gamma) \log n > n$, we trivially obtain the result with a suffix tree;
therefore the construction time of the grid is in $O(n\log n)$.


\begin{theorem} \label{thm:count2}
Let $T[1\dd n]$ have an attractor of size $\gamma$. Then, there exists a data
structure of size $O(\gamma\log(n/\gamma)\log n)$ that can count the number of
occurrences of any pattern $P[1\dd m]$ in $T$ in time $O(m)$.
The structure can be built in $O(n\log n)$ expected time and 
$O(n)$ space, without the need to know $\gamma$.
\end{theorem}

If we know $\gamma$, then an index that is correct w.h.p.\ can be built in 
$O(g \log n)$ space apart from the passes on $T$, but we must build $C''$ without using a suffix tree, in
additional time $O(\gamma \log^2 n \log^2(n/\gamma))$.
Table~\ref{tab:count} summarizes the results.

\begin{table}[t]
\begin{center}
\tbl{Space-time tradeoffs for counting; formulas are slightly simplified.
\label{tab:count}}
{\begin{tabular}{l|c|c}
Source & Space & Time \\
\hline
Baseline \cite{Nav18}
        & $O(z\log(n/z))$
        & $O(m\log^{2+\epsilon} n)$ \\
\hline
Theorem~\ref{thm:count}
        & $O(\gamma\log(n/\gamma))$
        & $O(m+\log^{2+\epsilon} n)$ \\
Theorem~\ref{thm:count2}
        & $O(\gamma\log(n/\gamma)\log n)$
        & $O(m)$ \\
\hline
\end{tabular}}
\end{center}
\end{table}


\section{Conclusions}

The size $\gamma$ of the smallest string attractor of a text $T[1..n]$ is a recent measure of compressibility \cite{KP18} that is particularly well-suited to express the amount of information in repetitive text collections. It asymptotically lower-bounds many other popular dictionary-based compression measures like the size $z$ of the Lempel--Ziv parse or the size $g$ of the smallest context-free grammar generating (only) $T$, among many others. It is not known whether one can always represent $T$ in compressed form in less than $\Theta(\gamma\log(n/\gamma))$ space, but within this space it is possible to offer direct access and reasonably efficient searches on $T$ \cite{KP18,NP18}.

In this article we have shown that, within $O(\gamma\log(n/\gamma))$ space, one can offer much faster searches, in time competitive with, and in most cases better than, the best existing results built on other dictionary-based compression measures, all of which use $\Omega(z\log(n/z))$ space. By building on the measure $\gamma$, our results immediately apply to any index that builds on other dictionary measures like $z$ and $g$. Our results are even competitive with self-indexes based on statistical compression, which are much more mature: we can locate the $occ$ occurrences in $T$ of a pattern $P[1..m]$ in $O(m+(occ+1)\log^\epsilon n)$ time, and count them in $O(m+\log^{2+\epsilon} n)$ time, whereas the fastest statistically-compressed indexes obtain $O(m+occ \log^\epsilon n)$ time to locate and $O(m)$ time to count, in space proportional to the statistical entropy of $T$ \cite{Sad03,BN13}.

Further, we show that our results can be obtained without even knowing an attractor nor its minimum size $\gamma$. Rather, we can compute a lower bound $\delta \le \gamma$ in linear time and use it to achieve $O(\gamma\log(n/\gamma))$ space without knowing $\gamma$. This is relevant because computing $\gamma$ is NP-hard \cite{KP18}. Previous work \cite{NP18} assumed that, although they obtained indexes bounded in terms of $\gamma$, one would compute some upper bound on it, like $z$, to apply it in practice. With our result, we obtain results bounded in terms of $\gamma$ without the need to find it.

Finally, we also obtain for the first time optimal search time using any index bounded by a dictionary-based compression measure. Within space $O(\gamma\log(n/\gamma)\log^\epsilon n)$, for any constant $\epsilon>0$, we can locate the occurrences in time $O(m+occ)$, and within $O(\gamma\log(n/\gamma)\log n)$ space we can count them in time $O(m)$. This is an important landmark, showing that it is possible to obtain the same optimal time reached by suffix trees in $O(n)$ space, now in space bounded in terms of a very competitive measure of repetitiveness. Such optimal time had also been obtained within space bounded by other measures that adapt to repetitiveness \cite{GNP18,BC17}, but these are weaker than $\gamma$ both in theory and in practice. Further, no statistical-compressed self-index using $o(n)$ space has obtained such optimal time.

As a byproduct, our developments yield a number of new or improved results on accessing and indexing on RLCFGs and CFGs; these are collected in Appendix~\ref{sec:rlcfgs}.

\paragraph{Future work}
There are still several interesting challenges ahead:

\begin{itemize}
\item While one can compress any text $T$ to $O(z)$ or $O(g)$ space (and even to smaller measures like $O(b)$ \cite{SS82}), it is not known whether one can compress it to $o(\gamma\log(n/\gamma))$ space. This is important to understand the nature of the concept of attractor and of measure $\gamma$.
\item While one can support direct access and searches on $T$ in space $O(g)$, it is not known whether one can support those in $o(z\log(n/z))$ or $o(\gamma\log(n/\gamma))$ space. Again, determining if this is a lower bound would yield a separation between $\gamma$, $z$, and $g$ in terms of indexability.
\item If we are given the size $\gamma$ of some attractor, we can build our indexes in a streaming-like mode, with $O(\log(n/\gamma))$ expected passes on $T$ plus main-memory space bounded in terms of $\gamma$, with high probability. This is relevant in practice when indexing huge text collections. It would be important to do the same when no bound on $\gamma$ is known. Right now, if we do not know $\gamma$, we need $O(n)$ extra space for a suffix tree that computes the measure $\delta \le \gamma$.
\item It is not clear if we can reach optimal search time in the ``minimum'' space $O(\gamma\log(n/\gamma))$, or what is the best time we can obtain in this case.
\item The measure $\delta$ is interesting on its own, as it lower-bounds
$\gamma$. It is interesting to find more precise bounds in terms of $\gamma$, and whether we can compress $T$, and even offer direct access and indexed searches on
it, within space $O(\delta\log(n/\delta))$.
\item The fact that only $O(\log m)$ partitions of $P$ are needed to spot all of its occurrences, which outperforms previous results \cite{NIIBT15,DBLP:conf/soda/GawrychowskiKKL18}, was fundamental to obtain our bounds, and we applied them to counting in order to obtain optimal times as well. It is likely that this result is of even more general interest and can be used in other problems related to dictionary-compressed indexing and beyond.
\item The result we obtain on counting pattern occurrences in $O(\gamma\log(n/\gamma))$ space is generalized to CFGs in Appendix~\ref{sec:rlcfgs}, but we could not generalize our result on specific RLCFGs to arbitrary ones. It is open whether this is possible or not.
\end{itemize}

\section*{Acknowledgements}
We thank Travis Gagie for pointing us the early reference related to $\delta$ \cite{RRRS13} and to Dmitry Kosolobov, who pointed out that a referenced result holds for constant alphabets only \cite{gagie2014lz77}.

\bibliographystyle{acmsmall}
\bibliography{paper}

\appendix

\newcommand{\T}{\mathcal T}
\newcommand{\F}{\mathcal F}

\section{Some Results on Run-Length Context-Free Grammars}
\label{sec:rlcfgs}

Along the article we have obtained a number of results for the specific RLCFG we build. Several of those can be generalized to arbitrary RLCFGs, leading to the same state of the art that CFGs now enjoy. We believe it is interesting to explicitly state those new results in general form: not only RLCFGs are always smaller than CFGs (and the difference can be asymptotically relevant, as in text $T = a^n$), but also our results in this article require space $O(\gamma\log(n/\gamma))$, whereas there always exists a RLCFG of size $g_{rl} = O(\gamma\log(n/\gamma))$. Indexes of size $O(g_{rl})$ have then the potential to be smaller than those built on attractors (e.g., $T=a^n$ is generated by a RLCFG of size $O(1)$, whereas $\gamma\log(n/\gamma) = O(\log n)$).

\subsection{Extracting substrings}
\label{sec:rlcfg-extract}

The following result exists on CFGs \cite{BLRSRW15}. They present their result on straight-line programs (SLPs, i.e., CFGs where right-hand sides are two nonterminals or one terminal symbol).
While any CFG of size $g$ can be converted into an SLP of size $O(g)$, we start by describing their structure generalized to arbitrary CFGs, which may be interesting when the grammar cannot be modified for some reason. We then show how to handle run-length rules $A \rightarrow A_1^s$
in order to generalize the result to RGCFGs.

\begin{theorem} \label{thm:rlcfg-extract}
Let a RLCFG of size $g_{rl}$ generate (only) $T[1\dd n]$. Then there exists
a data structure of size $O(g_{rl})$ that extracts any substring $T[p\dd 
p+\ell-1]$ in time $O(\ell+\log n)$.
\end{theorem}

Consider the parse tree $\T$ of $T[1\dd n]$. A {\em heavy path} starting at a 
node $v \in \T$ with children $v_1,\ldots,v_s$
chooses the child $v_i$ that maximizes $|v_i|$, and continues by $v_i$ in the 
same way, up to reaching a leaf. We say that $v_i$ is the {\em heavy child} of
$v$ and define $h(v)=v_i$. The edge connecting $v$ with its heavy child $v_i$ is
said to be {\em heavy}; those connecting $v$ with its other children are 
{\em light}. Note that, if $v_j \neq h(v)$, then $|v_j| \le |v|/2$; otherwise 
$v_j$ would be the heavy child of $v$. Then, every time we descend by a 
light
edge, the lenght of the node halves, and as a consequence no path from the root
to a leaf may include more than $\log n$ light edges. A decomposition into heavy
paths consists of the heavy path starting at the root of $\T$ and, recursively,
all those starting at the children by light edges.

\subsubsection{Accessing $T[p]$}

For every internal node $v$ with children $v_1,\ldots,v_s$ we define the 
starting positions of its children as $p_1(v)=1$, $p_i(v)=p_{i-1}(v)+|v_{i-1}|$,
for $2 \le i \le s$, and $p_{s+1}=|v|+1$. We then store the set $C(v) = \{ 
p_1(v), p_2(v), \ldots, p_{s+1}(v) \}$. Let us define $c(v)=p_i(v)$, where 
$v_i=h(v)$, as the starting position of the heavy child of $v$. Then, if $v$
roots a heavy path $v=v^0, v^1, \ldots, v^k$, where $v^j=h(v^{j-1})$ for 
$1 \le j \le k$, and $v^k$ is a leaf, we define
the starting positions in the heavy path as $s_1(v)=c(v)$ and 
$s_j(v)=s_{j-1}(v)-1+c(v^{j-1})$ for $2 \le j \le k$, and the ending 
positions as $e_j(v)=s_j(v)+|v^j|$ for $1 \le j \le k$. We then associate
with $v$ the increasing set $P(v) = \{ s_1(v), s_2(v), \ldots, s_k(v), e_k(v), 
\ldots, e_2(v), e_1(v) \}$; note $e_k(v)=s_k(v)+1$.

To find $T[p]$, we start at the root $v$ of $\T$ (so $1 \le p 
\le |v|$) with children $v_1,\ldots,v_s$. We make a predecessor search on $C(v)$
to determine that $p_i(v) \le p < p_{i+1}(v)$. If $v_i \neq h(v)$, we traverse 
the light edge to $v_i$ and continue the search from $v_i$ with $p 
\leftarrow p - p_i(v) + 1$. Otherwise, since $v_i=h(v)$, it holds that
$p \ge p_i(v) = c(v) = s_1(v)$ and $p < p_{i+1}(v) = c(v)+|h(v)| = e_1(v)$.
We then jump to the proper node in the heavy path that starts in $v$ by making
a predecessor search for $p$ in $P(v)$. If we determine that $s_j(v) \le p
< s_{j+1}(v)$ or that $e_{j+1}(v) \le p < e_j(v)$, we continue the search from 
$v^j$ and $p \leftarrow p-s_j(v)+1$. Otherwise, $p=s_k(v)$ and the answer
is the terminal symbol associated with the leaf $v^k$. Note that, when we
continue from $v^j$, this is not the head of a heavy path, but after searching
$C(v^j)$ we are guaranteed to continue by a light edge. In each step, then, we
perform two predecessor searches and traverse a light edge.

\citeN{BLRSRW15} describe a predecessor data structure that, when
finding the predecessor of $x$ in a universe of size $u$, takes time
$O(\log(u/(x^+-x^-)))$, where $x^+$ and $x^-$ are the predecessor and 
successor of $x$, respectively. Thus, when finding $v_i$ in $C(v)$, 
this structure takes time $O(\log(|v|/|v_i|))$. If $v_i$ is a light child, we
continue by $v_i$, so the sum over all the light edges traversed telescopes to
$O(\log |v|)$. When we descend to the heavy child, instead, we also find the
node $v^j$ in $P(v)$, which costs $O(\log(|v|/(s_{j+1}(v)-s_j(v)+1)))
= O(\log(|v|/c(v^j)))$ if $s_j(v) \le p < s_{j+1}(v)$, or 
$O(\log(|v|/(e_j(v)-e_{j+1}(v)+1))) = O(\log(|v|/(|v^j|-(c(v^j)+
|h(v^j)|))))$
if $e_{j+1}(v) \le p < e_j(v)$, or $O(\log |v|)$ if $p=s_k(v)$ (but this 
happens only once along the search). In the first two cases, we descend to 
$v^j$, which always starts
descending by a light edge to some $v^j_i$ at cost $O(\log(|v^j|/|v^j_i|))$.
Since $|v^j_i| \le c(v^j)$ (if $s_j(v) \le p < s_{j+1}(v)$) or
$|v^j_i| \le |v^j|-(c(v^j)+ |h(v^j)|)$ (if $e_{j+1}(v) \le p < e_j(v)$),
we can upper bound the cost to search $P(v)$ by $O(\log(|v|/|v^j_i|))$, and the
cost to search $C(v^j)$ by $O(\log(|v^j|/|v^j_i|)) \subseteq O(\log(|v|/|v^j_i|))$ too, and then we continue
the search from $v^j_i$. Therefore the cost also telescopes to $O(\log |v|)$
when we search a heavy path. Overall, the cost from the root of the parse tree
is $O(\log n)$.

The remaining problem is that the structure is of size $O(|\T|)=O(n)$, 
but it can be made $O(g)$ as follows. The subtrees of $\T$ rooted by all the
nodes $v$ labeled with the same nonterminal $A$ are identical, so in all of
them the node $h(v)$ has the same label, say the terminal or nonterminal $A_i$.
\citeN{BLRSRW15} define a forest $\F$ with exactly one node $v(X) 
\in \F$ for each nonterminal or nonterminal $X$. If $v \in \T$ is labeled $A$ 
and $h(v) \in \T$ is labeled $A_i$, then $v(A_i)$ is the parent of $v(A)$ in $\F$. 
The nodes $v(a)$ for terminals $a$ are roots in $\F$. A heavy path from
$v \in \T$, with $v$ labeled $A$, then corresponds to an upward path 
from $v(A) \in \F$.

The sets $C(v)$ also depend only on the label $A$ of $v \in \T$, so we 
associate them to the corresponding nonterminal $A$. The sizes of all sets
$C(A)$ add up to the grammar size, because $C(A)$ has $s+1$ elements
if the rule that defines $A$ is of the form $A \rightarrow A_1 \cdots A_s$.\footnote{To have the grammar size count only right-hand sides, rules $A \rightarrow \varepsilon$ must be removed or counted as size 1.}
The sets $P(v)$ also depend only on the label $A$ of $v \in \T$, but they are
not stored completely in $A$. Instead, each node $v(A) \in \F$,
corresponding to the nodes $v\in \T$ labeled $A$, and with parent $v(A_i) \in \F$,
stores values $s(v(A))=s(v(A_i))+c(v)-1$ and $e(v(A))=e(v(A_i))+|v|-c(v)-|h(v)|+1$.
For the roots $v(a) \in F$, we set $s(v(a)) = e(v(a)) = 0$. They then build two 
data structures for predecessor queries on tree paths, one on the $s(\cdot)$
and one on the $e(\cdot)$ values, which obtain the same complexities as on 
arrays. In order to find a position $p$ from $v(A)$, we also store the position
$p(A)$ in $\\exp(A)$ of the root in $\F$ from where $v(A)$ descends, as well as
the character $\\exp(A)[p(A)]$. If $p=p(A)$, we just return that symbol and 
finish. Otherwise, if $p<p(A)$, we search for $p(A)-p$ in the fields $s(\cdot)$
from $v(A)$ to the root, finding $s(v(B)) \ge p(A)-p > s(v(B_i))$, with $v(B_i)$ the
parent of $v(B)$ in $\mathcal{F}$. 
Otherwise, $p>p(A)$ and we search for $p-p(A)$ in the fields 
$e(\cdot)$ from $v(A)$ to the root, finding $e(v(B)) \ge p-p(A) > e(v(B_i))$, 
with $v(B_i)$ the parent of $v(B)$ in $\mathcal F$. In both cases, we must exit 
the heavy path from the node $v(B)$, adjusting $p \leftarrow p-s(v(A))+s(v(B))$.

\subsubsection{Extracting $T[p\dd q]$}

To extract $T[p\dd q]$ in time $O(q-p+\log n)$, we store additional information 
as follows. In each heavy path $v^0,\ldots,v^k$, each node $v^j$ stores a 
pointer $r(v^j) = h(v^t)$, where $j < t \le k$ is the smallest value for which 
$h(v^t)$ is not the rightmost child of $v^t$. Similarly, $l(v^j) = h(v^t)$
for the smallest $j < t \le k$ for which $h(v^t) > 1$. At query time, we apply
the procedures to retrieve $T[p]$ and $T[q]$ simultaneously until they split
at a node $v^*$, where $T[p]$ descends from the child $v^*_i$ and $T[q]$ from
the child $v^*_j$. Then the symbols $T[p\dd q]$ are obtained by traversing, in
left-to-right order, $(1)$ the children $v_{i+1},\ldots$ of every light edge
leading to $v_i$ in the way to $T[p]$; $(2)$ every sibling to the right of 
$r(v)$ for the nodes $v \in \{ v_1, r(v_1), r(r(v_1)), \ldots \}$ for every 
$v_1$ rooting a heavy path in the way to $T[p]$; $(3)$ the children 
$\{ v^*_{i+1},\ldots,v^*_{j-1}\}$ of $v^*$; $(4)$ the children 
$v_1,\ldots,v_{i-1}$ of every light edge $v_i$ in the way to $T[q]$; $(5)$
every sibling to the left of $l(v)$ for the nodes $v \in \{ v_1, l(v_1), 
l(l(v_1)), \ldots \}$ for every $v_1$ rooting a heavy path in the way to $T[q]$.
For all those nodes, we traverse their subtrees completely to obtain chunks of
$T[p\dd q]$ in optimal time (unless there are unary paths in the grammar, which
can be removed or skipped with the information on $r(\cdot)$ or $l(\cdot)$). 
The left-to-right order between nodes in $(1)$ and
$(2)$, and in $(3)$ and $(4)$, is obtained as we descend to $T[p]$ or $T[q]$.
Finally, $v^*$ is easily determined if it is the target of a light edge. 
Otherwise, if we exit a heavy path by distinct nodes $v_p$ and $v_q$, then
$v^*$ is the highest of the two.

\subsubsection{Extending to RLCFGs}

The idea to include rules $A \rightarrow A_1^s$ is to handle them exactly as if
they were $A \rightarrow A_1 \cdots A_1$, but using $O(1)$ space instead of $O(s)$.
When $v$ is labeled $A$ and this is defined as $A \rightarrow A_1^s$, we would 
have a tie in determining the heavy child $h(v)$. We then act as if we chose 
the first copy of $A_1$, $h(v)=v_1$; in particular $v(A_1)$ is the parent of $v(A)$
in $\F$. If we have to descend by another child of $v$ to reach position $p$ 
inside $v$, we choose $v_i$ with $i = \lceil p/|v_1| \rceil$ and set 
$p \leftarrow p - (i-1)\cdot|v_1|$, so we do not need to store the set
$C(A)$ (which would exceed our space budget).

No pointer $l(v^j)$ will point to $h(v)$, but pointers $r(v^j)$ will.
The pointers $r(v^j) = h(v^t)$ are actually stored as a pair $(v^t,i)$ where
$v^s_i = h(v^t)$; this allows accessing preceding and following siblings
easily. With this format, we can also refer to the $i$th child of a 
run-length node and handle it appropriately.

\subsection{Extracting prefixes and suffixes}

The following result also exists on CFGs \cite{DBLP:conf/dcc/GasieniecKPS05}, who use leftmost or rightmost paths instead of heavy paths. In our Lemma~\ref{lem:extract from rlcfg} we have extended it to arbitrary RLCFGs as well, without setting any restriction on the grammar.

\begin{theorem} \label{thm:rlcfg-prefsuf}
Let a RLCFG of size $g_{rl}$ generate (only) $T[1\dd n]$. Then there exists
a data structure of size $O(g_{rl})$ that extracts any prefix or suffix of
the expansion $\exp(A)$ of any nonterminal $A$ in real time.
\end{theorem}

\subsection{Computing fingerprints}

The following result, already existing on CFGs \cite{BGCSVV17}, can also be extended to arbitrary RLCFGs. Note that it improves our Lemma~\ref{lem:kr} to $O(\log\ell)$ time, though we opted for a simpler variant in the body of the article.

\begin{theorem} \label{thm:rlcfg-kr}
Let a RLCFG of size $g_{rl}$ generate (only) $T[1\dd n]$. Then there exists
a data structure of size $O(g_{rl})$ that computes the Karp-Rabin signature
of any substring $T[p\dd q]$ in time $O(\log n)$.
\end{theorem}

Recall that, given the signatures $\kappa(S_1)$ and $\kappa(S_2)$, one can compute
the signature of the concatenation, $\kappa(S_1 \cdot S_2) =
(\kappa(S_1) + c^{|S_1|} \cdot \kappa(S_2)) \bmod \mu$. One can also compute
the signature of $S_2$ given those of $S_1$ and $S_1 \cdot S_2$,
$\kappa(S_2) = ((\kappa(S_1 \cdot S_2) - \kappa(S_1)) \cdot c^{-|S_1|})
\bmod \mu$, and the signature of $S_1$ given those of $S_2$ and $S_1 \cdot S_2$,
$\kappa(S_1) = (\kappa(S_1 \cdot S_2) - \kappa(S_2) \cdot c^{|S_1|}) \bmod \mu$.
To have the terms $c^{\pm |S_1|}$ handy, we redefine signatures $\kappa(S)$ as
triples $(\kappa(S),c^{|S|} \bmod \mu,c^{-|S|} \bmod \mu)$, which are 
easily maintained across the described operations.

We now show how to compute a fingerprint $\kappa(T[p\dd q])$ in $O(\log n)$
time on an arbitrary RLCFG. We present the current result \cite{BGCSVV17},
extended to general CFGs, and then include run-length rules.

We follow the idea of our Lemma~\ref{lem:kr}, but combine it with heavy paths.
Since we can obtain $\kappa(T[p\dd q])$ from $\kappa(T[1\dd q])$ and 
$\kappa(T[1\dd p-1])$, we only consider computing fingerprints of text prefixes.
We associate with each nonterminal $A \rightarrow A_1 \cdots A_s$ the $s$
signatures $K_i(A) = \kappa(\exp(A_1)\cdots \exp(A_{i-1}))$, for $1 \le i \le s$.
We also associate signatures to nodes $v(A)$ in $\mathcal F$,
$K(v(A)) = \kappa(\exp(A)[1\dd p(A)-1]$. Those values fit in $O(g)$ 
space. 

To compute $\kappa(T[1\dd p])$ we start with $\kappa = 0$ and follow the same 
process as for accessing $T[p]$ in Section~\ref{sec:rlcfg-extract}. In our way, every time we descend by a light 
edge from $v$ to $v_i$, where $v$ is labeled $A$, we update $\kappa \leftarrow 
(\kappa + K_i(A) \cdot c^{|A_1|+\cdots+|A_{i-1}|}) \bmod \mu$. Note that the
power of $c$ is implicitly stored together with the signature $K_i(A)$ itself.

Instead, when we descend from $v(A)$ to $v(B)$ because $s(v(B)) \ge p(A)-p >
s(v(B_i))$ or $e(v(B)) \ge p-p(A) > e(v(B_i))$, we first compute the signature 
$\kappa'$ of the prefix of $\exp(A)$ that precedes $\exp(B)$, which is of length 
$\ell = s(v(A))-s(v(B))$, and then update $\kappa \leftarrow \kappa\cdot c^\ell
+ \kappa'$ so as to concatenate that prefix (again, $c^\ell$ is computed 
together with $\kappa'$).
We compute $\kappa'$ from $K(v(B)) = \kappa(\exp(B)[1\dd p(B)-1])$ and
$K(v(A)) = \kappa(\exp(A)[1\dd p(A)-1])$. Because $\exp(A)[p(A)]$ is the
same symbol of $\exp(B)[p(B)]$, $\exp(B)[1\dd p(B)-1]$ is a suffix of
$\exp(A)[1\dd p(A)-1]$. We then use the method to extract $\kappa(S_1)$ from
$\kappa(S_1 \cdot S_2)$ and $\kappa(S_2)$. 

When we arrive at $T[p]$, we include that symbol and have computed
$\kappa = \kappa(T[1\dd p])$. The time is the same $O(\log n)$ required to
access $T[p]$.

\subsubsection{Handling run-length rules}
The proof of Lemma~\ref{lem:kr} already shows how to handle run-length rules $A \rightarrow A_1^s$: we again treat them as 
$A \rightarrow A_1 \cdots A_1$. The only complication is that now we cannot afford
to store the values $K_i(A)$ used to descend by light edges, but we can compute
them as $K_i(A) = \kappa(\exp(A_1)^{i-1}) =
\left(\kappa(\exp(A_1))\cdot \frac{c^{|A_1|\cdot (i-1)}-1}{c^{|A_1|}-1}\right) \bmod 
\mu$: $c^{|A_1|} \!\!\mod \mu$ and $(c^{|A_1|}-1)^{-1} \bmod \mu$ can be 
stored with $A_1$, and the exponentiation can be computed in time
$O(\log i) \subseteq O(\log s)$. Note
that this is precisely the $O(\log(|v|/|v_i|))$ time we are allowed to
spend when moving from node $v$ to its child $v_i$ by a light edge.

\subsection{Locating pattern occurrences}

\citeN[Cor.~1]{CNspire12} obtain a version of the following result that holds only for CFGs and offers search time $O(m^2+(m+occ)\log^\epsilon n)$. We improve their complexity and generalize it to RLCFGs.

\begin{theorem} \label{thm:rlcfg-index}
Let a RLCFG of size $g_{rl}$ generate (only) $T[1\dd n]$. Then there exists
a data structure of size $O(g_{rl})$ that finds the $occ$ occurrences in 
$T$ of any pattern $P[1\dd m]$ in time
$O(m\log n + occ \log^\epsilon n)$ for any constant $\epsilon>0$.
\end{theorem}

This result is essentially obtained in our Section~\ref{sec:index}. In that section we use a specific RLCFG that allows us obtain a better complexity. However, in a general RLCFG, where we must search for all the $\tau=m-1$ possible splits of $P$, the application of Lemma~\ref{lemma: z-fast} with complexities $f_e(\ell) = O(\ell)$ (Theorem~\ref{thm:rlcfg-prefsuf}) and $f_h(\ell)=O(\log n)$ (Theorem~\ref{thm:rlcfg-kr}) yields $O(m\log n)$ time to find all the $m-1$ ranges $[x_1\dd x_2] \times [y_1\dd y_2]$ to search for in the grid.

Combining that result with the linear-space grid representation and the mechanism to track the secondary occurrences on the grammar tree of a RLCFG described in Section~\ref{sec:index}, the result follows immediately.

\subsection{Counting pattern occurrences}

While we cannot generalize our result of Section~\ref{sec:counting} to arbitrary RLCFGs, our developments let us improve the best current result on arbitrary CFGs \cite{Nav18}.

\begin{theorem} \label{thm:rlcfg-count}
Let a CFG of size $g$ generate (only) $T[1\dd n]$. Then there exists
a data structure of size $O(g)$ that computes the number of occurrences in 
$T$ of any pattern $P[1\dd m]$ in time
$O(m\log^{2+\epsilon} n)$ for any constant $\epsilon>0$.
\end{theorem}

\citeN[Thm.~4]{Nav18} showed that the number of times $P[1\dd m]$ occurs in 
$T[1\dd n]$ can be computed in time $O(m^2 + m\log^{2+\epsilon} n)$ within 
$O(g)$ space for any CFG of size $g$. As explained in Section~\ref{sec:counting}, he uses the same grid of our Section~\ref{sec:index} for the primary
occurrences, but associates with each point the number of occurrences 
triggered by it (which depend only on the point). Then, a linear-space 
geometric structure \cite{DBLP:journals/siamcomp/Chazelle88} sums all the numbers in a range in time
$O(\log^{2+\epsilon} g)$. Adding over all the $m-1$ partitions of $P$,
and considering the $O(m^2)$ previous time to find all the ranges
\cite{CNspire12}, the final complexity is obtained.

With Lemma~\ref{lemma: z-fast}, and given our new results in 
Theorems~\ref{thm:rlcfg-prefsuf} and \ref{thm:rlcfg-kr}, we can now improve Navarro's result 
to $O(m\log^{2+\epsilon} n)$ because the $O(m^2)$ term becomes $O(m\log n)$.
However, this holds only for CFGs. Run-length rules introduce significant
challenges, in particular the number of secondary occurrences do not 
depend only on the points.
We only could handle this issue for the specific RLCFG we use in Section~\ref{sec:counting}. An interesting open problem is to generalize this solution to arbitrary RLCFGs.

\end{document}